\documentclass[orivec]{llncs}

\usepackage{sty/header}
\usepackage{sty/abbreviations}
\usepackage{sty/symbols}
\usepackage{sty/defs}
\usepackage{sty/versioning}

\fullfalse
\withappendixtrue

\anonymousfalse

\begin{document}

\title{Quasi-Optimal Partial Order Reduction\thanks{This paper is the extended
version of a paper with the same title appeared at the proceedings of CAV'18.}}

\ifanonymous
	\author{}
	\institute{}
\else
	\author{Huyen T.T. Nguyen\inst{1}
		\and César Rodríguez\inst{1,3}
		\and Marcelo Sousa\inst{2}
		\and\\ Camille Coti\inst{1}
		\and Laure Petrucci\inst{1}}
	
	\institute{Université Paris 13, Sorbonne Paris Cité, CNRS, France\\
		\and
		University of Oxford, United Kingdom
		\and
		Diffblue Ltd. Oxford, United Kingdom}
\fi

\maketitle 
\pagenumbering{arabic}

\begin{abstract}
A dynamic partial order reduction (DPOR) algorithm is optimal when it always
explores at most one representative per Mazurkiewicz trace. 
Existing literature suggests that the reduction obtained by the non-optimal,
state-of-the-art Source-DPOR (SDPOR) algorithm is comparable to optimal DPOR.
We show the first program\footnote{Shortly after this extended version was made
public, we were made
aware of the recent publication of another paper~\cite{AAJS17} which contains an
independently-discovered example program with the same characteristics.}
with~$\bigo n$ Mazurkiewicz traces where 
SDPOR explores $\bigo{2^n}$ redundant schedules.
We furthermore identify the cause of this blow-up as an NP-hard problem.
Our main contribution is a new approach, called Quasi-Optimal POR, 
that can arbitrarily approximate an optimal exploration 
using a provided constant~$k$. 
We present an implementation of our method 
in a new tool called~\dpu~using specialised data structures.
Experiments with~\dpu, including Debian packages, show that optimality
is achieved with low values of~$k$, outperforming state-of-the-art tools. 
\end{abstract}

\section{Introduction}
\label{sec:intro}

Dynamic partial-order reduction (DPOR)~\cite{FG05,AAJS14,RSSK15} is a mature 
approach to mitigate the state explosion problem in stateless model checking of 
multithreaded programs.
DPORs are based on Mazurkiewicz trace theory~\cite{Maz87}, a true-concurrency 
semantics where the set of executions of the program is partitioned into
equivalence classes known as Mazurkiewicz traces (M-traces).
In a DPOR, this partitioning is defined by an independence relation over
concurrent actions that is computed dynamically and the method
explores executions which are representatives of M-traces.
The exploration is \emph{sound} when it explores all M-traces, 
and it is considered \emph{optimal}~\cite{AAJS14} when it 
explores each M-trace only once. 

Since two independent actions
might have to be explored from the same state 
in order to explore all M-traces, 
a DPOR algorithm uses independence to
compute a provably-sufficient subset of the enabled 
transitions to explore for each state encountered.
Typically this involves the combination of
forward reasoning (persistent sets~\cite{God96} or
source sets~\cite{AAJS14,AAJS17b}) with backward reasoning 
(sleep sets~\cite{God96}) to obtain a more efficient
exploration.
However, in order to obtain optimality, 
a DPOR is forced to compute sequences of 
transitions (as opposed to sets of enabled transitions)
that avoid visiting a previously visited M-trace.
These sequences are stored 
in a data structure called~\emph{wakeup trees} in~\cite{AAJS14}
and known as~\emph{alternatives} in~\cite{RSSK15}.
Computing these sequences thus amounts to deciding whether the DPOR
needs to visit yet another M-trace (or all have already been seen).

In this paper, we prove that computing alternatives 
in an optimal DPOR is an NP-complete problem.
To the best our knowledge
this is the first formal complexity 
result on this important subproblem that
optimal and non-optimal DPORs need to solve.
The program shown in~\cref{fig:ssbs}~(a)
illustrates a practical consequence of this result:
the non-optimal, state-of-the-art SDPOR algorithm~\cite{AAJS14}
can explore here~$\bigo{2^n}$ interleavings but the program has only~$\bigo n$
M-traces.
\begin{figure*}[t]
\centering
\hspace{-30mm}
\begin{tabular}[b]{l}
\scalebox{0.6}{
	\begin{tabular}[b]{l}
		\parbox[b]{16cm}{
			\setlength{\tabcolsep}{15pt}
			\hspace{3pt}
			(a)~
			\large
			\begin{tabular}{ccccc}
				$w_0$ & $w_1$ & $w_2$ & count & master \\
				\hline
				$x_0=7$ & $x_1=8$ & $x_2=9$ & $c=1$ & $i=c$ \\
				& & & $c=2$ & $x_i=0$
			\end{tabular}
				
		} 
      \\[20pt]
		\parbox[b]{12cm}{
		\scalebox{0.85}{
\begin{tikzpicture}
[node distance=1cm, auto, every node/.style={rectangle, draw, text centered, minimum size=5mm}, every
path/.style={->,>=stealth'},
plain/.style={draw=none, text centered}]
\normalsize

\node at (0,1) [label=above:{$w_0$}] (1) {$1$};
\node [right of=1,label=above:{$w_1$}] (2) {$2$};
\node [right of=2,label=above:{$w_3$}] (3) {$3$};
\node [right of=3,label=above:{$i=0$}] (4) {$4$};
\node at (1,-0.25) [label=left:{$x_0=0$}] (5) {$5$};
\node at (3,0) [label=left:{$c=1$}] (6) {$6$};
\node at (3,-1) [label=left:{$c=2$}] (7) {$7$};
\draw (1) -- (5);
\draw (4) -- (5);
\draw (4) -- (6);
\draw (6) -- (7);
\node at (-0.5,-1.5) [plain] (b) {\large (b)};

\node at (4.5,-1) [label=below:{$w_0$}] (8) {$8$};
\node at (4.5,1) [label=above:{$w_1$}] (2b) {$2$};
\node [right of=2b,label=above:{$w_2$}] (3b) {$3$};
\node [right of=3b,label=above:{$i=0$}] (4b) {$4$};
\node at (5,0) [label=left:{$x_0=0$}] (5b) {$5$};
\node at (6.5,0) [label=left:{$c=1$}] (6b) {$6$};
\node at (6.5,-1) [label=left:{$c=2$}] (7b) {$7$};
\draw (4b) -- (5b);
\draw (4b) -- (6b);
\draw (6b) -- (7b);
\draw (5b) -- (8);

\node at (8,1) [label=above:{$w_1$}] (1c) {$1$};
\node [right of=1c,label=above:{$w_2$}] (2c) {$2$};
\node [right of=2c,label=above:{$w_3$}] (3c) {$3$};
\node at (11,0) [label=left:{$i=1$}] (10) {$10$};
\node at (9,-1) [label=above left:{$x_1=0$}] (11) {$11$};
\node at (11,1) [label=above:{$c=1$}] (9) {$9$};
\node [below of=10,label=left:{$c=2$}] (12) {$12$};
\draw (2c) -- (11);
\draw (9) -- (10);
\draw (10) -- (11);
\draw (10) -- (12);

\node at (12.5,1) [label=above:{$w_0$}] (1e) {$1$};
\node [right of=1e,label=above:{$w_1$}] (2e) {$2$};
\node [right of=2e,label=above:{$w_2$}] (3e) {$3$};
\node at (14.5,-0.75) [label=right:{$i=2$}] (15) {$15$};
\node at (13,-1.5) [label=right:{$x_2=0$}] (16) {$16$};
\node at (15.5,1) [label=above:{$c=1$}] (9e) {$9$};
\node at (15.5,0) [label=left:{$c=2$}] (14) {$14$};
\draw (9e) -- (14);
\draw (14) -- (15);
\draw (15) -- (16);
\draw (3e) -- (16);

\node at (17,1) [label=above:{$w_1$}] (1f) {$1$};
\node [right of=1f,label=above:{$w_1$}] (2f) {$2$};
\node at (18.5,-1.5) [label=right:{$w_2$}] (17) {$17$};
\node at (18,-0.25) [label=above:{$i=2$}] (15f) {$15$};
\node at (17,-0.75) [label=below:{$x_2=0$}] (16f) {$16$};
\node at (19,1) [label=above:{$c=1$}] (9f) {$9$};
\node at (19,0) [label=below:{$c=2$}] (14f) {$14$};
\draw (9f) -- (14f);
\draw (14f) -- (15f);
\draw (15f) -- (16f);
\draw (16f) -- (17);

\node at (20.5,1) [label=above:{$w_0$}] (1d) {$1$};
\node at (21.5,-1.5) [label=above:{$w_1$}] (13) {$13$};
\node at (21.5,1) [label=above:{$w_2$}] (3d) {$3$};
\node at (22.5,0) [label=left:{$i=1$}] (10d) {$10$};
\node at (20.5,-0.5) [label=above:{$x_1=0$}] (11d) {$11$};
\node at (22.5,1) [label=above:{$c=1$}] (9d) {$9$};
\node [below of=10d,label=below:{$c=2$}] (12d) {$12$};
\draw (9d) -- (10d);
\draw (10d) -- (11d);
\draw (10d) -- (12d);
\draw (11d) -- (13);

\end{tikzpicture}}
		}	
	\end{tabular}
}
\end{tabular}
\centering
\caption{(a): Programs; (b): Partially-ordered executions;}
\label{fig:ssbs}
\vspace{-15pt}
\end{figure*}
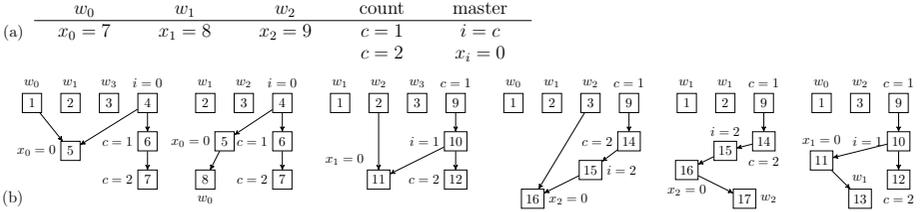

The program 
contains $n \eqdef 3$~\emph{writer} threads $w_0, w_1, w_2$, each writing
to a different variable.
The thread~\emph{count} increments $n-1$ times a zero-initialized counter~$c$.
Thread~\emph{master} reads $c$ into variable~$i$ and writes to~$x_i$.

The statements \mbox{$x_0 = 7$} and $x_1 = 8$ are independent because 
they produce the same state regardless of their execution order.
Statements $i = c$ and any statement in the~\emph{count} thread are dependent 
or~\emph{interfering}: their execution orders result in different states.
Similarly, $x_i = 0$ interferes with exactly one \emph{writer} thread, 
depending on the value of~$i$.

Using this independence relation, the set of executions of this program 
can be partitioned into six~M-traces, corresponding to the 
six~partial orders shown in~\cref{fig:ssbs}~(b).
Thus, an optimal DPOR explores six~executions 
($2n$-executions for~$n$ \emph{writers}).
We now show why SDPOR explores~$\bigo{2^n}$ in the general case. 
Conceptually, SDPOR is a loop that
(1) runs the program,
(2) identifies two dependent statements that can be swapped, and
(3) reverses them and re-executes the program.
It terminates when no more dependent statements can be swapped.

Consider the interference on the counter variable~$c$ between the~\emph{master}
and the~\emph{count} thread.
Their execution order determines which~\emph{writer} thread interferes with
the~\emph{master} statement~$x_i= 0$.
If $c=1$ is executed just before $i=c$, then $x_i=0$ interferes with~$w_1$.
However, if $i=c$ is executed before, then $x_i=0$ interferes with~$w_0$.
Since SDPOR does not track relations between dependent statements,
it will naively try to 
reverse the race between $x_i=0$ and \emph{all writer threads}, which 
results in exploring~$\bigo{2^n}$ executions.
In this program, exploring only six traces requires understanding the 
entanglement between both interferences as the order in which the first is
reversed determines the second.

As a trade-off solution between solving this NP-complete problem and 
potentially explore an exponential number of redundant schedules, 
we propose a hybrid approach called Quasi-Optimal POR (QPOR) which 
can turn a non-optimal DPOR into an optimal one.
In particular, we provide a polynomial algorithm to compute alternative
executions that can arbitrarily approximate the optimal solution
based on a user specified constant~$k$.
The key concept is a new notion of \emph{$k$-partial alternative}, 
which can intuitively be seen as a ``good enough'' alternative:
they revert two interfering statements while remembering the 
resolution of the last~$k-1$ interferences.

The major differences between QPOR and the DPORs of~\cite{AAJS14} are that:
1) QPOR is based on prime event structures~\cite{NPW80}, a partial-order
semantics that has been recently applied to
programs~\cite{RSSK15,SRDK17}, instead of a sequential view to thread
interleaving, and
2) it computes~$k$-partial alternatives with an $\bigo{n^k}$ algorithm
while optimal DPOR corresponds to computing $\infty$-partial 
alternatives with an~$\bigo{2^n}$ algorithm.
\remove{
Similar to SDPOR, QPOR is a loop that
(1) executes the program until the end and transforms the execution into a
dependency graph,
(2) finds the next dependency graph to explore by computing the $k$-partial 
alternative to the current dependency graph, \ie, a set of events in the event 
structure that conflicts with at least~$k$ events in the dependency graph, and
(3) re-executes the program following the new dependency graph found in (2).

Computing $k$-partial alternatives requires only $\bigo{n^k}$ time.
%
The original DPOR algorithm~\cite{FG05} and~SDPOR correspond to (a simplified 
version of) a POR that employs $1$-partial alternatives, \ie, at any given time 
in the exploration their scope is restricted to one race.

Optimal POR corresponds to $\infty$-partial alternatives.
Computing $\infty$-partial alternatives is NP-complete,
but computing $k$-partial alternatives, with $k \in \N$
requires $\bigo{n^k}$ time.
}
For the program shown in~\cref{fig:ssbs}~(a), QPOR achieves optimality
with~$k=2$ because races are coupled with (at most) another race.
As expected, the cost of computing~$k$-partial alternatives and the reductions 
obtained by the method increase with higher values of~$k$.

Finding $k$-partial alternatives requires decision procedures 
for traversing the causality and conflict relations in event structures.
Our main algorithmic contribution is to represent these relations as a set of 
trees where events are encoded as one or two nodes in two different trees.
We show that checking causality/conflict between events amounts 
to an efficient traversal in one of these trees.

In summary, our main contributions are:
\begin{itemize}[noitemsep,topsep=0pt]
\item
  Proof that computing alternatives for optimal DPOR is NP-complete~(\cref{sec:complexity}).
\item
  Efficient data structures and algorithms for
  (1) computing $k$-partial alternatives in polynomial time, and
  (2) represent and traverse partial orders~(\cref{sec:algo}).
\item 
  Implementation of QPOR in a new tool called \dpu and 
  experimental evaluations against SDPOR in~\nidhugg 
  and the testing tool~\maple~(\cref{sec:exp}).

\item
  Benchmarks with $\bigo{n}$ M-traces where SDPOR
  explores~$\bigo{2^n}$ executions (\cref{sec:exp}).
\end{itemize}

Furthermore, in~\cref{sec:exp} we show that:
(1) low values of~$k$ often achieve optimality; 
(2) even with non-optimal explorations \dpu greatly outperforms \nidhugg;
(3) \dpu copes with production code in Debian packages and achieves much higher
state space coverage and efficiency than~\maple.

Proofs for all our formal results are available in the appendix
of this manuscript.

\section{Preliminaries}
\label{sec:prelim}

In this section we provide the formal background used throughout the paper.

\paragraph{Concurrent Programs.}

We consider deterministic concurrent programs composed of a fixed number of
threads that communicate via shared memory and
synchronize using mutexes
(\cref{fig:ssbs}~(a) can be trivially modified to satisfy this). We also assume that
local statements can only modify shared memory within a mutex block.  Therefore, it
suffices to only consider races of mutex accesses.

Formally, a \emph{concurrent program} is a structure
$P \eqdef \tup{\mem, \locks, T, m_0, l_0}$, where
$\mem$~is the set of~\emph{memory states}
(valuations of program variables, including instruction pointers),
$\locks$~is the set of~\emph{mutexes},
$m_0$~is the~\emph{initial memory state}, 
$l_0$~is the~\emph{initial mutexes state} and
$T$ 
is the set of \emph{thread statements}.
A thread statement $t \eqdef \tup{i,f}$ is a pair where~$i \in \nat$
is the~\emph{thread identifier}
associated with the statement and~$f \colon \mem \to (\mem \times \act)$
is a \emph{partial} function that models the transformation of the memory as
well as the \emph{effect}
$\act \eqdef \set \local \cup (\set{\lock,\unlock} \times \locks)$
of the statement with respect to thread synchronization.
Statements of~$\local$ effect model local thread code.
Statements associated with $\tup{\lock,x}$ or $\tup{\unlock,x}$ model lock and 
unlock operations on a mutex~$x$.
Finally, we assume that
(1) functions~$f$ are PTIME-decidable;
(2) $\lock$/$\unlock$ statements do not modify the memory; and
(3) $\local$ statements modify thread-shared memory only
within lock/unlock blocks.
When (3) is violated, then~$P$ has a \emph{datarace}
(undefined behavior in almost all languages),
and our technique can be used to find such statements, see~\cref{sec:exp}.

We use \emph{labelled transition systems} (\lts) semantics for our programs.
We associate a program~$P$ 
with the \lts~$M_P \eqdef \tup{\stat, {\to}, A, s_0}$.
The set $\stat \eqdef \mem \times (\locks \to \set{0,1})$ are the
\emph{states} of~$M_P$, \ie, pairs of the form~$\tup{m,v}$
where~$m$ is the state of the memory and~$v$ indicates when a mutex is
locked~(1) or unlocked~(0).
The \emph{actions} in~$A \subseteq \nat \times \act$ are pairs~$\tup{i,b}$
where~$i$ is the identifier of the thread that executes some statement and~$b$
is the effect of the statement.
We use the function~$p \colon A \to \nat$ to retrieve the thread identifier.
The \emph{transition relation}~${\to} \subseteq \stat \times A \times \stat$
contains a triple
$\tup{m,v} \fire{\tup{i,b}} \tup{m',v'}$
exactly when there is some thread statement $\tup{i,f} \in T$ such
that~$f(m) = \tup{m',b}$ and either
(1) $b = \local$ and $v' = v$, or
(2) $b = \tup{\lock,x}$ and $v(x) = 0$ and $v' = v_{|x\mapsto 1}$, or
(3) $b = \tup{\unlock,x}$ and $v' = v_{|x\mapsto 0}$.
Notation $f_{x\mapsto y}$ denotes a function that behaves like~$f$ for all
inputs except for~$x$, where $f(x) = y$.
The \emph{initial state} is~$s_0 \eqdef \tup{m_0,l_0}$.

Furthermore,
if $s \fire{a} s'$ is a transition, 
the action~$a$ is \emph{enabled} at~$s$.
Let~$\enabl s$ denote the set of actions enabled at~$s$.
A sequence $\sigma \eqdef a_1 \ldots a_n \in A^*$ is a \emph{run} when
there are states $s_1, \ldots, s_n$ satisfying
$s_0 \fire{a_1} s_1 \ldots \fire{a_n} s_n$.
We define $\statee \sigma \eqdef s_n$.
We let $\runs{M_P}$ denote the set of all runs 
and
$\reach{M_P} \eqdef \set{\statee \sigma \in \stat \colon \sigma \in \runs{M_P}}$
the set of all \emph{reachable states}.

\paragraph{Independence.}

Dynamic partial-order reduction methods use a notion called independence to
avoid exploring concurrent interleavings that lead to the same state.
We recall the standard notion of independence for actions in~\cite{God96}.
Two actions~$a, a' \in A$ \emph{commute at} a state~$s \in \stat$ iff
\begin{itemize}[noitemsep,topsep=0pt]
\item
  if $a \in \enabl s$ and $s \fire a s'$,
  then $a' \in \enabl s$ iff $a' \in \enabl{s'}$; and
\item
  if $a, a' \in \enabl s$,
  then there is a state $s'$ such that
  $s \fire{a.a'} s'$ and
  $s \fire{a'.a} s'$.
\end{itemize}
Independence between actions is an under-approximation of commutativity.
A binary relation ${\indep} \subseteq A \times A$ is an
\emph{independence} on~$M_P$ if it is symmetric, irreflexive, and 
every pair $\tup{a, a'}$ in $\indep$ commutes at every state in $\reach{M_P}$.

In general $M_P$ has multiple independence relations,
clearly $\emptyset$ is always one of them.
We define
relation~${\indep_P} \subseteq A \times A$ as the smallest irreflexive, symmetric
relation where
$\tup{i,b} \indep_P \tup{i',b'}$ holds if $i \neq i'$ and either
$b = \local$ or
$b = \lock\ x$ and $b' \not \in \set {\lock\ x, \unlock\ x}$.
By construction~$\indep_P$ is always an independence.

\paragraph{Labelled Prime Event Structures.}
\emph{Prime event structures} (\pes) are well-known non-interleaving, partial-order
semantics~\cite{NPW81,EH08,Esp10spin}.
\iffull
Standard dynamic partial-order reduction methods~\cite{FG05,AAJS14} are based 
on the interleaving semantics of the associated \lts~$M_P$, 
\ie, they reason over executions of the program.
QPOR will operate over a non-interleaving, partial-order based semantics,
known as prime event structures.
\fi
Let $X$ be a set of actions.
A \emph{\pes over~$X$} is a structure $\les \eqdef \tup{E, <, \cfl, h}$ where
$E$ is a set of \emph{events},
${<} \subseteq E \times E$ is a strict partial order
called \emph{causality relation},
${\cfl} \subseteq E \times E$ is a symmetric, irreflexive
\emph{conflict relation},
and $h \colon E \to X$ is a labelling function.
Causality represents the happens-before relation between events,
and conflict between two events expresses that any execution includes at most
one of them.
\cref{fig:exunf}~(b) shows a \pes over $\nat \times \act$
where causality is depicted by arrows,
conflicts by dotted lines, and the labelling~$h$ is shown next to the events,
\eg, $1 < 5$, $8 < 12$, $2 \cfl 8$, and $h(1) = \tup{0,\local}$.
The \emph{history} of an event $e$,
$\causes e \eqdef \set{e' \in E \colon e' < e}$,
is the least set of events that need to happen before~$e$.

The notion of concurrent execution in a~$\pes$ is captured by the concept of~\emph{configuration}.
A configuration is a (partially ordered) execution of the system, \ie,
a set~$C \subseteq E$
of events that is \emph{causally closed}
(if $e \in C$, then $\causes e \subseteq C$)
and \emph{conflict-free}
(if $e, e' \in C$, then $\lnot (e \cfl e')$).
In \cref{fig:exunf}~(b), the set $\set{8,9,15}$ is a configuration,
but $\set{3}$ or $\set{1,2,8}$ are not.
We let $\conf \les$ denote the set of all configurations of~$\les$, and
$[e] \eqdef \causes e \cup \set e$ the \emph{local configuration} of~$e$.
In \cref{fig:exunf}~(b), $[11] = \set{1,8,9,10,11}$.
A configuration represents a set of \emph{interleavings} over~$X$.
An interleaving is a sequence in~$X^*$ that labels any topological sorting of
the events in~$C$.
We denote by $\inter C$ the set of interleavings of~$C$.
In \cref{fig:exunf}~(b),
$\inter{\set{1,8}} = \set{ab, ba}$ with
$a \eqdef \tup{0,\local}$ and
$b \eqdef \tup{1,\lock\ m}$.

The \emph{extensions} of $C$ 
are the events not in~$C$ whose histories are included in~$C$:
$\ex C \eqdef \set{e \in E \colon e \notin C \land \causes e \subseteq C}$.
The \emph{enabled} events of~$C$ are the extensions that can form a
larger configuration:
$\en C \eqdef \set{e \in \ex C \colon C \cup \set e \in \conf{\les}}$.
%
Finally, the \emph{conflicting extensions} of~$C$ are the extensions that are
not enabled: $\cex C \eqdef \ex C \setminus \en C$.
In \cref{fig:exunf}~(b),
$\ex{\set{1,8}} = \set{2,9,15}$,
$\en{\set{1,8}} = \set{9,15}$, and
$\cex{\set{1,8}} = \set{2}$.
See \cite{RSSK15long} for more information on \pes concepts.

\paragraph{Parametric Unfolding Semantics.}

We recall the program \pes semantics of~\cite{RSSK15,RSSK15long}
(modulo notation differences).
For a program $P$ and any independence~$\indep$ on~$M_P$ we define a
\pes~$\unf{P,\indep}$
that represents the behavior of~$P$, \ie, such that
the interleavings of its set of configurations equals $\runs{M_P}$.
%

Each event in~$\unf{P,\indep}$ is defined by a canonical name
of the form $e \eqdef \tup{a,H}$, where $a \in A$ is an action of~$M_P$
and~$H$ is a configuration of $\unf{P,\indep}$.
Intuitively, $e$ represents the action~$a$
after the \emph{history} (or the causes) $H$.
\cref{fig:exunf}~(b) shows an example.
Event~11 is $\tup{\tup{0, \lock\ m},\set{1,8,9,10}}$ and
event~1 is $\tup{\tup{0, \local}, \emptyset}$.
Note the inductive nature of the name, and how it
allows to uniquely identify each event.
We define the \emph{state of a configuration} as the state reached by \emph{any}
of its interleavings.
Formally, for $C \in \conf{\unf{P,\indep}}$ we define
$\statee C$ as $s_0$ if $C = \emptyset$ and
as $\statee \sigma$ for some $\sigma \in \inter C$ if $C \ne \emptyset$.
Despite its appearance $\statee C$ is well-defined because \emph{all} sequences
in $\inter C$ reach the \emph{same} state, see~\cite{RSSK15long} for a proof.
\begin{definition}[Unfolding]
\label[definition]{def:unfd}
Given a program~$P$ and some independence relation $\indep$
on $M_P \eqdef \tup{\stat, \to, A, s_0}$, the
\emph{unfolding of~$P$ under~$\indep$}, denoted $\unf{P,\indep}$,
is the \pes over~$A$ constructed by the following fixpoint rules:
\begin{enumerate}[topsep=0pt]
\item
  Start with a \pes $\les \eqdef \tup{E, <, {\cfl}, h}$
  equal to $\tup{\emptyset, \emptyset, \emptyset, \emptyset}$.
\item
  Add a new event $e \eqdef \tup{a,C}$ to~$E$ for any
  configuration $C \in \conf \les$ and any action $a \in A$ if
  $a$~is enabled at $\statee C$ and
  $\lnot (a \indep h(e'))$ holds for every $<$-maximal event $e'$ in~$C$.
\item
  For any new $e$ in $E$, update $<$, $\cfl$, and $h$ as follows:
    for every $e' \in C$, set $e' < e$;
    for any $e' \in E \setminus C$,
    set $e' \cfl e$
    if $e \ne e'$ and $\lnot (a \indep h(e'))$;
    set $h(e) \eqdef a$.
\item
  Repeat steps 2 and 3 until no new event can be added to~$E$;
  return $\les$.
\end{enumerate}
\end{definition}
Step 1 creates an empty \pes with only one (empty) configuration.
Step 2 inserts a new event $\tup{a,C}$ by finding a configuration~$C$ that
enables an action~$a$ which is dependent with all causality-maximal events
in~$C$.
In \cref{fig:exunf},
this initially creates events~1, 8, and 15.
For event $1 \eqdef \tup{\tup{0,\local}, \emptyset}$, this is because
action $\tup{0,\local}$ is enabled at $\statee \emptyset = s_0$
and there is no $<$-maximal event in~$\emptyset$ to consider.
Similarly, the state of
$C_1 \eqdef \set{1, 8, 9, 10}$ enables action
$a_1 \eqdef \tup{0, \lock\ m}$, and both~$h(1)$ and~$h(10)$ are
dependent with~$a_1$ in~$\indep_P$.
As a result $\tup{a_1,C_1}$ is an event (number 11).
Furthermore, while $a_2 \eqdef \tup{0, \local}$ is enabled at $\statee{C_2}$,
with $C_2 \eqdef \set{8,9,10}$,
$a_2$ is independent of $h(10)$ and $\tup{a_2,C_2}$ is not an event.

After inserting an event~$e \eqdef \tup{a,C}$, \cref{def:unfd} declares
all events in~$C$ causal predecessors of~$e$.
For any event~$e'$ in $E$ but not in $[e]$ such that~$h(e')$ is dependent
with~$a$, the order of execution of~$e$ and~$e'$ yields different states.
We thus
set them in conflict.
In \cref{fig:exunf},
we set $2 \cfl 8$ because
$h(2)$ is dependent with~$h(8)$ and
$2 \notin [8]$ and $8 \notin [2]$.

\begin{figure}[t]
\centering
\hspace{-5pt}
\begin{minipage}[t]{5.4cm}
\footnotesize

\itshape{Thread 0:\qquad Thread 1:\quad Thread 2:}
\begin{verbnobox}[\ttfamily]
x := 0      lock(m)    lock(m')
lock(m)     y := 1     z := 3
if (y == 0) unlock(m)  unlock(m')
  unlock(m)
else
  lock(m')
  z := 2
\end{verbnobox}
\normalfont (a)
\end{minipage}
\begin{minipage}[t]{6cm}
\raisebox{-3.75cm}{
	\scalebox{0.6}{
\begin{tikzpicture}
\begin{scope}
[node distance=1cm, auto, every node/.style={rectangle, draw, text centered, minimum size=5mm}, 
every path/.style={->,>=stealth'},
plain/.style={draw=none, text centered}]
\normalsize

\node at (0,1) [label=left:{$\tup{0,\local}$}] (1) {$1$};
\node [below of=1,label=left:{$\tup{0,\lock\ m}$}] (2) {$2$};
\node [below of=2,label=left:{$\tup{0,\local}$}] (3) {$3$};
\node [below of=3,label=left:{$\tup{0,\unlock\ m}$}] (4) {$4$};
\node at (1,-2.75) [label=left:{$\tup{1,\lock\ m}$}] (5) {$5$};
\node [below of=5, label=left:{$\tup{1,\local}$}] (6) {$6$};
\node [below of=6, label=left:{$\tup{1, \unlock\ m}$}] (7) {$7$};
\draw (1) -- (2);
\draw (2) -- (3);
\draw (3) -- (4);
\draw (4) -- (5);
\draw (5) -- (6);
\draw (6) -- (7);

\node at (3.25,1) [label=right:{$\tup{1,\lock\ m}$}] (8) {$8$};
\node [below of=8,label=right:{$\tup{1,\local}$}] (9) {$9$};
\node [below of=9,label=right:{$\tup{1,\unlock\ m}$} ] (10) {$10$};
\node at (2.25,-1.75) [label=right:{$\tup{0,\lock\ m}$}] (11) {$11$};
\node [below of=11,label=right:{$\tup{0,\local}$}] (12) {$12$};
\node [below of=12, label=right:{$\tup{0,\lock\ m'}$}] (13) {$13$};
\node [below of=13, label=right:{$\tup{0,\local}$}] (14) {$14$};
\draw (8) -- (9);
\draw (9) -- (10);
\draw (1) -- (11);
\draw (10) -- (11);
\draw (11) -- (12);
\draw (12) -- (13);
\draw (13) -- (14);

\node at (6,1) [label=right:{$\tup{2,\lock\ m'}$}] (15) {$15$};
\node at (6,-1) [label=right:{$\tup{2,\local}$}] (16) {$16$};
\node [below of=16,label=right:{$\tup{2,\unlock\ m'}$} ] (17) {$17$};
\node at (5.5,-3.75) [label=right:{$\tup{0,\lock\ m'}$}] (18) {$18$};
\node [below of=18,label=right:{$\tup{0,\local}$}] (19) {$19$};
\draw (15) -- (16);
\draw (16) -- (17);
\draw (12) -- (18);
\draw (17) -- (18);
\draw (18) -- (19);
\end{scope}

\node at (-1.75,-5) [draw=none]	(b)		{\Large (b)};

\begin{scope}
\path[dashed, thick, color=red] (2) edge (8);
\path[dashed, thick, color=red] (13) edge[bend right=27] (15);
\end{scope}
\end{tikzpicture}}
}
\end{minipage}
\caption{(a): a program $P$; (b): its unfolding semantics $\unf{P,\indep_P}$.}
\label{fig:exunf}
\end{figure}

\section{Unfolding-Based DPOR}
\label{sec:por}

\SetKwProg{Proc}{Procedure}{}{}
\SetKwProg{Fn}{Function}{}{}
\SetKwIF{If}{ElseIf}{Else}{if}{}{else if}{else}{end}
\SetKwFor{ForEach}{foreach}{}{end}
\SetKwFor{For}{for}{}{end}
\SetKwInOut{Input}{Input}
\SetKwInOut{Output}{Output}
\SetKw{Continue}{continue}
\SetKw{Break}{break}
\SetInd{0.2em}{1em}
\SetNlSty{}{\color{gray}}{}
\SetVlineSkip{2pt}

\SetKwFunction{explore}{Explore}
\SetKwFunction{alternatives}{Alt}
\SetKwFunction{exten}{ex}
\SetKwFunction{cexp}{cexp}
\SetKwFunction{pt}{pt}
\SetKwFunction{pm}{pm}
\SetKwFunction{remove}{Remove}

This section presents an
algorithm that exhaustively explores all deadlock states of a given
program (a \emph{deadlock} is a state where no thread is enabled).

For the rest of the paper, unless otherwise stated, we let $P$ be a
\emph{terminating} program
(\ie, $\runs{M_P}$ is a finite set of finite sequences) and
$\indep$ an independence on~$M_P$.
Consequently, $\unf{P,\indep}$ has finitely many events and configurations.

Our POR algorithm (\cref{a:a1}) analyzes~$P$ by exploring the configurations
of~$\unf{P,\indep}$.
It visits all $\subseteq$-maximal configurations of~$\unf{P,\indep}$,
which correspond to the deadlock states in~$\reach{M_P}$,
and organizes the exploration as a binary tree.

\explore{$C,D,A$} has a global set~$U$ that stores
all events of~$\unf{P,\indep}$ discovered so far.
The three arguments are:
$C$, the configuration to be explored;
$D$ (for \emph{disabled}), a set of events that shall never be visited (included
in~$C$) again; and
$A$ (for \emph{add}), used to direct the exploration towards a configuration
that conflicts with~$D$.
A call to \explore{$C,D,A$} visits all maximal configurations of~$\unf{P,\indep}$
which contain~$C$ and do not contain~$D$, and the first one explored
contains~$C \cup A$.

\begin{algorithm}[t]
\DontPrintSemicolon
\setlength{\columnsep}{2pt}
\begin{multicols}{2}
   Initially, set $U \eqdef \emptyset$,\\
   and call \explore{$\emptyset$, $\emptyset$, $\emptyset$}. \;

   \BlankLine

   \Proc{\explore{$C, D, A$}}{
      Add $\ex C$ to $U$ \; \label{l:ext}

      \lIf{$\en C \subseteq D$}{ \KwRet } \label{l:ret}
      
      \eIf{$A = \emptyset$}
      {
         Choose $e$ from $\en C \setminus D$ \; \label{l:a2choose}
      }
      {
         Choose $e$ from $A \cap \en C$ \; \label{l:a1choose}
      }

      \explore{$C \cup \set e, D, A \setminus \set e$} \;
      \label{l:left}
      \uIf{$\exists J \in \alternatives{$C, D \cup \set e$}$ \label{l:alt}} 
      {
         \explore{$C, D \cup \set e, J \setminus C$}
         \label{l:right}
      }
      $U \eqdef U \cap Q_{C,D}$ \; \label{l:a1remove}
   } 

\columnbreak
\Fn{\cexp{C}}{
   $R \eqdef \emptyset$ \;
   \ForEach {event $e \in C$ of type $\lock$}
   {
      $e_t \eqdef \pt(e)$ \;
      $e_m \eqdef \pm(e)$ \;
      \While {$\lnot (e_m \leq e_t)$ \label{l:lcau}}  
      { 
         $e_m \eqdef \pm(e_m)$ \label{l:em} \;
         \lIf {$(e_m < e_t)$}{\Break}
         $e_m \eqdef \pm(e_m)$ \;
         $\hat e \eqdef \tup{h(e), [e_t] \cup [e_m]}$ \;
         \label{l:hate}
         Add $\hat e$ to $R$
      }
   }
   \KwRet{R}   
}
\end{multicols}
\vspace*{6pt}
\caption{Unfolding-based POR exploration. See text for definitions.}
\label{a:a1}
\end{algorithm}

The algorithm first adds~$\ex C$ to~$U$.
If~$C$ is a maximal configuration (\ie, there is no enabled event) then 
\cref{l:ret} returns.
If~$C$ is not maximal but $\en C \subseteq D$, then all possible events that
could be added to~$C$ have already been explored and this call was redundant
work.
In this case the algorithm also returns and we say that it has explored
a~\emph{sleep-set blocked} (SSB) execution~\cite{AAJS14}.
\cref{a:a1} next selects an event enabled at~$C$, if possible from~$A$
(\cref{l:a2choose,l:a1choose}) and
makes a recursive call (left subtree) that explores
\emph{all} configurations that contain all events in $C \cup \set e$ and
no event from~$D$.
Since that call visits all maximal configurations 
containing~$C$ and~$e$, it remains to visit those 
containing~$C$ but not~$e$.
At \cref{l:alt} we determine if any such configuration exists.
Function \alternatives returns a set of configurations, so-called \emph{clues}.
A clue is a witness that a $\subseteq$-maximal configuration exists
in~$\unf{P,\indep}$ which contains $C$ and not $D \cup \set{e}$. 

\begin{definition}[Clue]
\label[definition]{def:clue}
Let $D$ and $U$ be sets of events, and
$C$ a configuration such that $C \cap D = \emptyset$.
A \emph{clue} to~$D$ after~$C$ 
in~$U$ is a configuration~$J \subseteq U$ 
such that $C \cup J$ is a configuration and $D \cap J = \emptyset$.
\end{definition}

\begin{definition}[\texttt{Alt} function]
\label[definition]{def:alt-fun}
Function \alternatives denotes \emph{any} function such that
\alternatives{$B,F$} returns a set of clues to~$F$ after~$B$ in~$U$, and
the set is non-empty if
$\unf{P,\indep}$ has at least one maximal configuration~$C$
where $B \subseteq C$ and $C \cap F = \emptyset$.
\end{definition}

When \alternatives returns a clue $J$, the clue is passed in the second
recursive call (\cref{l:right}) to ``mark the way'' (using set~$A$)
in the subsequent recursive calls at \cref{l:left}, and guide the exploration
towards the maximal configuration that $J$ witnesses.
\cref{def:alt-fun} does not identify a concrete implementation of \alternatives.
It rather indicates how to implement \alternatives so that~\cref{a:a1}
terminates and is complete (see below).
Different PORs in the literature can be reframed in terms of~\cref{a:a1}.
SDPOR~\cite{AAJS14} uses clues that mark the way with only one event ahead ($|J
\setminus C| = 1$) and can hit SSBs.
Optimal DPORs~\cite{AAJS14,RSSK15} use size-varying clues that guide the
exploration provably guaranteeing that any SSB will be avoided.

\cref{a:a1} is \emph{optimal} when it does not explore a SSB.
To make \cref{a:a1} optimal \alternatives needs to return clues that are
\emph{alternatives}~\cite{RSSK15}, which satisfy stronger constraints.
When that happens, \cref{a:a1} is equivalent to the DPOR
in~\cite{RSSK15} and becomes optimal (see \cite{RSSK15long} for a proof).

\begin{definition}[Alternative~\cite{RSSK15}]
\label[definition]{def:alternative}
Let~$D$ and~$U$ be sets of events and~$C$ a configuration such that
$C \cap D = \emptyset$.
An \emph{alternative} to~$D$ after~$C$ in~$U$ is a clue~$J$ to~$D$ after~$C$
in~$U$ such that $\forall e \in D: \exists e' \in J$, $e \cfl e'$.
\end{definition}

\Cref{l:a1remove} removes from~$U$ events that will not
be necessary for~\alternatives to find clues in the future.
The events preserved, $Q_{C,D} \eqdef C \cup D \cup \fcfl{C \cup D}$,
include all events in~$C \cup D$ as well as every event in~$U$ that
is in conflict with some event in~$C \cup D$.
The preserved events will suffice to compute alternatives~\cite{RSSK15},
but other non-optimal implementations of~\alternatives could allow
for more aggressive pruning.

The $\subseteq$-maximal configurations of \cref{fig:exunf}~(b) are
$[7] \cup [17]$,
$[14]$, and
$[19]$.
Our algorithm starts at configuration $C = \emptyset$.
After~10 recursive calls it visits $C = [7] \cup [17]$.
Then it backtracks to~$C = \set{1}$, calls \alternatives{$\set{1}, \set{2}$},
which provides, \eg, $J = \set{1,8}$, and
visits~$C = \set{1,8}$ with~$D = \set{2}$.
After~6 more recursive calls it visits~$C = [14]$,
backtracks to $C = [12]$,
calls \alternatives{$[12], \set{2,13}$}, which provides, \eg,
$J = \set{15}$, and after two more recursive calls it
visits~$C = [12] \cup \set{15}$ with~$D = \set{2,13}$.
Finally, after~4 more recursive calls it visits~$C = [19]$.

Finally, we focus on the correctness of~\cref{a:a1}, and prove termination and
soundness of the algorithm:

\begin{restatable}[Termination]{theorem}{thmtermination}
\label{thm:ter}
Regardless of its input, \cref{a:a1} always stops.
\end{restatable}

\begin{restatable}[Completeness]{theorem}{thmcompleteness}
\label{thm:sou}
Let $\hat C$ be a $\subseteq$-maximal configuration of $\unf{P,\indep}$.
Then \cref{a:a1} calls \explore{$C,D,A$} at least once with $C=\hat C$.
\end{restatable}

\section{Complexity}
\label{sec:complexity}

This section presents complexity results about the only non-trival steps
in~\cref{a:a1}: computing~$\ex C$ and the call to~\alternatives{$\cdot,\cdot$}.
An implementation of \alternatives{$B,F$} that systematically returns~$B$ would
satisfy~\cref{def:alt-fun},
but would also render~\cref{a:a1} unusable (equivalent to a DFS in~$M_P$).
On the other hand the algorithm becomes optimal when \alternatives returns
alternatives. Optimality comes at a cost:

\begin{restatable}{theorem}{pesaltnp}
\label{thm:pes-altnp}
Given a finite \pes $\les$, some configuration
$C \in \conf \les$, and a set~$D \subseteq \ex C$, deciding if an
alternative to~$D$ after~$C$ exists in~$\les$ is NP-complete.
\end{restatable}

\Cref{thm:pes-altnp} assumes that $\les$ is an arbitrary \pes.
Assuming that $\les$ is the unfolding of a program~$P$ under~$\indep_P$
does not reduce this complexity:

\begin{restatable}{theorem}{progaltnp}
\label{thm:prog-altnp}
Let $P$ be a program and
$U$ a causally-closed set of events from $\unf{P,\indep_P}$.
For any configuration
$C \subseteq U$ and any $D \subseteq \ex C$,
deciding if an alternative to~$D$ after~$C$ exists in~$U$
is NP-complete.
\end{restatable}

These complexity results lead us to consider (in next section) new approaches
that avoid the NP-hardness of computing alternatives while still retaining their
capacity to prune the search.

Finally, we focus on the complexity of computing $\ex C$, which essentially
reduces to computing $\cex C$, as computing $\en C$ is trivial.
Assuming that $\les$ is given, computing $\cex C$ for some $C \in \conf \les$ is
a linear problem. However, for any realistic implementation of~\cref{a:a1},
$\les$ is not available (the very goal of~\cref{a:a1} is to find all of its
events).
So a useful complexity result about~$\cex C$ necessarily refers to the orignal
system under analysis.
When $\les$ is the unfolding of a Petri net~\cite{Mcm93}
(see \cref{sec:basic-defs} for a formal definition), computing~$\cex C$ is
NP-complete:
\begin{restatable}{theorem}{thmcexnp}
	\label{thm:cex-np}
	Let $N$ be a Petri net,
	$t$ a transition of~$N$,
   $\les$ the unfolding of~$N$
	and~$C$ a configuration of~$\les$.
	Deciding if $h^{-1}(t) \cap \cex C = \emptyset$ is NP-complete.
\end{restatable}

Fortunately, computing $\cex C$ for programs 
is a much simpler task.
Function \cexp{$C$}, shown in~\cref{a:a1}, computes and returns $\cex C$ when
$\les$ is the unfolding of some program.
We explain \cexp{$C$} in detail in~\cref{sec:cexp}.
But assuming that functions
\pt and \pm can be computed in constant time, and relation $<$ decided in
$\bigo{\log |C|}$, as we will show,
clearly \cexp works in time~$\bigo{n^2\log n}$, where $n \eqdef |C|$,
as both loops are bounded by the size of~$C$.

\section{New Algorithm for Computing Alternatives}
\label{sec:algo}

This section introduces a new class of clues, called $k$-partial alternatives.
These can arbitrarily reduce the number of redundant explorations
(SSBs) performed by \cref{a:a1} and can be computed in polynomial time.
Specialized data structures and algorithms for $k$-partial alternatives are also
presented.

\begin{definition}[k-partial alternative]
\label[definition]{def:kalt}
Let~$U$ be a set of events, $C \subseteq U$~a configuration, $D \subseteq U$
a set of events, and $k \in \N$ a number.
A configuration~$J$ is a \emph{$k$-partial alternative}
to~$D$ after~$C$ if there is some $\hat D \subseteq D$
such that $|\hat D| = k$ and~$J$ is an alternative to~$\hat D$ after~$C$.
\end{definition}

A $k$-partial alternative needs to conflict with only $k$ (instead of all)
events in~$D$.
An alternative is thus an $\infty$-partial alternative.
If we reframe SDPOR in terms of~\cref{a:a1}, it becomes an algorithm
using \emph{singleton 1-partial} alternatives.
While $k$-partial alternatives are a very simple concept,
most of their simplicity stems from the fact that they are expressed
within the elegant framework of \pes semantics.
Defining the same concept on top of sequential semantics
(often used in the~POR literature~\cite{God96,FG05,YCGK08,AAJS14,AAAJLS15,FHRV13}),
would have required much more complex device.

We compute $k$-partial alternatives using
a $comb$ data structure:
\begin{definition}[Comb]
\label[definition]{def:cmb}
Let~$A$ be a set.
An \emph{A-comb}~$c$ of size $n \in \N$ is an ordered collection of
\emph{spikes} $\tup{s_1,\ldots,s_n}$,
where each spike $s_i \in A^*$ is a sequence of elements over~$A$.
Furthermore, a \emph{combination} over $c$ is any tuple $\tup{a_1, \ldots, a_n}$ where
$a_i \in s_i$ is an element of the spike.
\end{definition}

It is possible to compute $k$-partial alternatives (and by extension optimal
alternatives) to~$D$ after~$C$ in~$U$ using a comb, as follows:

\begin{enumerate}[topsep=2pt]
\item
  Select $k$ (or $|D|$, whichever is smaller) arbitrary events
  $e_1, \ldots, e_k$ from~$D$.
\item
  Build a $U$-comb $\tup{s_1, \ldots, s_k}$ of size $k$, where spike~$s_i$
  contains all events in $U$ in 
  conflict with $e_i$.
\item
  Remove from~$s_i$ any event $\hat e$ such that either $[\hat e] \cup C$ is not a
  configuration or $[\hat e] \cap D \ne \emptyset$.
\item
  Find combinations $\tup{e'_1, \ldots, e'_k}$ in the comb satisfying 
  $\lnot (e'_i \cfl e'_j)$ for $i \ne j$.
\item
  For any such combination the set $J \eqdef [e'_1] \cup \ldots \cup [e'_k]$ is a
  $k$-partial alternative.
\end{enumerate}

Step 3 guarantees that~$J$ is a clue.
Steps~1 and~2 guarantee that it will conflict with at least~$k$ events from~$D$.
%
It is straightforward to prove that
the procedure will find a $k$-partial
alternative to~$D$ after~$C$ in~$U$ when an $\infty$-partial alternative to~$D$
after~$C$ exists in~$U$.
It can thus be used to implement \cref{def:alt-fun}.

Steps 2, 3, and 4 require to decide whether a given pair of events is
in conflict.
Similarly, step 3 requires to decide if two events are causally related.
Efficiently computing $k$-partial alternatives thus reduces to efficiently
computing causality and conflict between events.

\subsection{Computing Causality and Conflict for PES events}
\label{sub:cctree}

In this section we introduce an efficient data structure for deciding whether
two events in the unfolding of a program are causally related or in conflict.

As in \cref{sec:por}, let
$P$ be a program,
$M_P$ its LTS semantics, and
$\indep_P$ its independence relation (defined in~\cref{sec:prelim}).
Additionally, let $\les$ denote the \pes
$\unf{P,\indep_P}$ of~$P$ extended with a
new event~$\bot$ that causally precedes every event in~$\unf{P,\indep_P}$.

The unfolding~$\les$ represents the dependency of actions in~$M_P$ through
the causality and conflict relations between events.
By definition of~$\indep_P$ we know that for any two events $e,e' \in \les$:
\begin{itemize}
\item
  If $e$ and $e'$ are events from the same thread, then they are either causally
  related or in conflict.
\item
  If $e$ and $e'$ are lock/unlock operations on the same variable, then
  similarly they are either causally related or in conflict.
\end{itemize}

This means that the causality/conflict relations between all events of one thread
can be tracked using a tree.
For every thread of the program we define and maintain a so-called \emph{thread tree}.
Each event of the thread has a corresponding node in the tree.
A tree node~$n$ is the parent of another tree node~$n'$ iff the event associated
with~$n$ is the immediate causal predecessor of the event associated with~$n'$.
That is, the ancestor relation of the tree encodes
the causality relations of events in the thread,
and the branching of the tree represents conflict.
Given two events $e, e'$ of the same thread we have that
$e < e'$ iff $\lnot (e \cfl e')$ iff
the tree node of $e$ is an ancestor of the tree node of~$e'$.

We apply the same idea to track causality/conflict between $\lock$ and $\unlock$
events.
For every lock $l \in \locks$ we maintain a separate \emph{lock tree}, containing
a node for each event labelled by either $\tup{\lock, l}$ or $\tup{\unlock, l}$.
As before, the ancestor relation in a lock tree encodes the causality
relations of all events represented in that tree.
Events of type $\lock$/$\unlock$ have tree nodes in both their lock and thread
trees.
Events for $\local$ actions are associated to only one node in the thread tree.

This idea gives a procedure to decide a causality/conflict query for two events
when they belong to the same thread or modify the same lock.
But we still need to decide causality and conflict for other events, \eg,
$\local$ events of different threads.
Again by construction of~$\indep_P$, the only source of conflict/causality
for events are the causality/conflict relations between the causal
predecessors of the two.
These relations can be summarized by keeping two mappings for each event:

\begin{definition}
Let $e \in E$ be an event of~$\les$.
We define the
\emph{thread mapping} $\tmaxx \colon E \times \nat \to E$ as the only
function that maps every pair $\tup{e,i}$ to the unique $<$-maximal event
from thread~$i$ in~$[e]$, or~$\bot$ if~$[e]$ contains no event from thread~$i$.
Similarly,
the \emph{lock mapping} $\lmaxx \colon E \times \locks \to E$ maps every
pair $\tup{e,l}$ to the unique $<$-maximal event~$e' \in [e]$ such that
$h(e')$ is an action of the form $\tup{\lock, l}$ or $\tup{\unlock, l}$,
or~$\bot$ if no such event exists in~$[e]$.
\end{definition}

The information stored by the thread and lock mappings enables us to decide
causality and conflict queries for arbitrary pairs of events:

\begin{restatable}{theorem}{procaus}
\label{thm:caus}
Let $e, e' \in \les$ be two arbitrary events from \resp threads $i$ and~$i'$,
with $i \ne i'$.
Then $e < e'$ holds iff $e \leqslant \tmax{e', i}$. And
$e \cfl e'$ holds iff there is some $l \in \locks$ such that
$\lmax{e, l} \cfl \lmax{e', l}$.
\end{restatable}	 

As a consequence of \cref{thm:caus}, deciding whether two events are
related by causality or conflict
reduces to deciding whether two nodes from the \emph{same} lock or thread tree
are ancestors.

\subsection{Computing Causality and Conflict for Tree Nodes}
\label{ss:stree}
 
This section presents an efficient algorithm to decide if two nodes of a tree
are ancestors. The algorithm is similar to a search in a skip list~\cite{Pugh89}.

Let $\tup{N, {\lessdot}, r}$ denote a tree, where~$N$ is a set of \emph{nodes},
${\lessdot} \subseteq N \times N$ is the \emph{parent relation}, and $r \in N$
is the root.
Let $d(n)$ be the depth of each node in the tree, with $d(r) = 0$.
A node~$n$ is an \emph{ancestor} of~$n'$ if it belongs to the only path
from~$r$ to~$n'$.
Finally, for a node $n \in N$ and some integer $g \in \nat$ such that $g \le
d(n)$ let $q(n,g)$ denote the unique ancestor~$n'$ of~$n$ such that
$d(n') = g$.

Given two \emph{distinct} nodes $n, n' \in N$, we need to
efficiently decide whether $n$ is an ancestor of~$n'$.
The key idea is that
if $d(n) = d(n')$, then the answer is clearly negative;
and if the depths are different and \wlogg $d(n) < d(n')$,
then we have that $n$ is an ancestor of~$n'$ iff nodes~$n$ and
$n'' \eqdef q(n',d(n))$ are the same node.

To find~$n''$ from~$n'$, a linear traversal of the branch starting from~$n'$
would be expensive for deep trees.
Instead, we propose to use a data structure similar to a skip list.
Each node stores a pointer to the parent node \emph{and} also a number of
pointers to ancestor nodes at distances $s^1, s^2, s^3, \ldots$,
where~$s \in \nat$ is a user-defined \emph{step}.
The number of pointers stored at a node~$n$ is equal to the number of trailing
zeros in the~$s$-ary representation of~$d(n)$.
For instance, for $s \eqdef 2$ a node at depth~$4$ stores 2~pointers (apart from
the pointer to the parent) pointing to the nodes at depth~$4-s^1 = 2$ and
depth~$4-s^2 = 0$.
Similarly a node at depth~12 stores a pointer to the ancestor (at depth~11) and
pointers to the ancestors at depths~10 and~8.
With this algorithm computing $q(n,g)$ requires traversing $\log(d(n) - g)$
nodes of the tree.

\subsection{Computing Conflicting Extensions}
\label{sec:cexp}

We now explain how function \cexp{$C$} in~\cref{a:a1} works.
A call to \cexp{$C$} constructs and returns all events in~$\cex C$.
The function works only when the \pes being explored is the unfolding of a
program~$P$ under the independence~$\indep_P$.

Owing to the properties of~$\unf{P,\indep_P}$, all events
in~$\cex C$ are labelled by $\lock$ actions.
Broadly speaking, this is because
only the actions from different threads that are co-enabled \emph{and}
are dependent create conflicts in~$\unf{P,\indep_P}$.
And this is only possible for~$\lock$ statements.
For the same reason, an event labelled by
$a \eqdef \tup{i, \tup{\lock, l}}$ exists in~$\cex C$ iff
there is some event $e \in C$ such that $h(e) = a$.

Function \cexp exploits these facts and the lock tree
introduced in~\cref{sub:cctree} to compute~$\cex C$.
Intuitively, it finds every event $e$ labelled by an $\tup{\lock, l}$ statement
and tries to ``execute'' it before the $\tup{\unlock, l}$ that happened
before~$e$ (if there is one). If it can, it creates a new event~$\hat e$ with
the same label as~$e$.

Function \pt{$e$} returns the only immediate causal predecessor of event~$e$ in
its own thread.
For an $\lock$/$\unlock$ event~$e$,
function \pm{$e$} returns the parent node of event~$e$ in its lock tree (or
$\bot$ if $e$ is the root).
So for an~$\lock$ event it returns a~$\unlock$ event, and for a~$\unlock$ event
it returns an~$\lock$ event.

\section{Experimental Evaluation}
\label{sec:exp}
We implemented QPOR in a new tool
called~\dpu
(\emph{Dynamic Program Unfolder}, available at
\url{https://github.com/cesaro/dpu/releases/tag/v0.5.2}).
\dpu is a stateless model checker for C programs with POSIX threading.
It uses the LLVM infrastructure to parse, instrument, and JIT-compile the
program, which is assumed to be data-deterministic.
%
%
It implements $k$-partial alternatives ($k$ is an input), optimal POR, and
context-switch bounding~\cite{CMK13}.

\dpu does not use data-races as a source of thread interference for POR.
It will not explore two execution orders for the two instructions that exhibit a
data-race.
However, it can be instructed to detect and report data races found during
the POR exploration.
When requested, this detection happens for a user-provided percentage of the
executions explored by POR.

\subsection{Comparison to SDPOR}

In this section we investigate the following experimental questions:
%
 (a) How does QPOR compare against SDPOR?
 (b) For which values of~$k$ do $k$-partial alternatives yield optimal exploration?

We use realistic programs that expose complex thread 
synchronization patterns including a job dispatcher, a multiple-producer
multiple-consumer scheme, parallel computation of $\pi$, and a thread pool.
Complex synchronizations patterns are frequent in these examples, including
nested and intertwined critical sections or conditional interactions between
threads based on the processed data,
and provide means to highlight the differences between POR approaches and drive
improvement.
Each program contains between~2 and~8 assertions, often ensuring
invariants of the used data structures.
All programs are safe and have between~90 and~200 lines of code.
We also considered the SV-COMP'17 benchmarks, but almost all of them contain very
simple synchronization patterns, not representative of more complex concurrent
algorithms. \cref{sec:svcomp17} provides the experimental data of this
comparison.
On these benchmarks QPOR and SDPOR perform an almost identical exploration,
both timeout on exactly the same instances, and both find exactly the same bugs.

\newcommand\newrow{\\[-1.0pt]}
\newcommand\param[1]{\scriptsize(#1)}
\newcommand\h{\rmfamily\bfseries}

\newcommand\cmidrules{
  \cmidrule(r){1-3}
  \cmidrule(r){4-5}
  \cmidrule(r){6-7}
  \cmidrule(r){8-9}
  \cmidrule(r){10-11}
  \cmidrule(r){12-14}
}

\begin{table*}[!t]
\scriptsize

\setlength\tabcolsep{3.5pt}
\def\sep{\hspace{15pt}}
\def\tinysep{\hspace{5pt}}
\centering
\parbox[c][8cm]{\linewidth}{
\tt
\begin{tabular}[t]{lrr@{\tinysep}rr@{\tinysep}rr@{\tinysep}rr@{\tinysep}rr@{\sep}rrr}	
\toprule
  \multicolumn{3}{l}{\rm Benchmark}
& \multicolumn{2}{l}{\rm \dpu (k=1)}  
& \multicolumn{2}{l}{\rm \dpu (k=2)}
& \multicolumn{2}{l}{\rm \dpu (k=3)} 
& \multicolumn{2}{l}{\rm \dpu (optimal)} 
& \multicolumn{3}{l}{\rm \nidhugg}

\newrow
\cmidrules

  {\rm Name}
& {\rm Th}
& {\rm Confs}

& {\rm Time}
& {\rm SSB}

& {\rm Time}
& {\rm SSB}

& {\rm Time}
& {\rm SSB}

& {\rm Time}
& {\rm Mem}

& {\rm Time}
& {\rm Mem}
& {\rm SSB}

\newrow
\midrule


\rm\sc Disp\param{5,2}    &     8 &      137 &    0.8   &     1K   &    0.4   &       43 &    0.4   &        0 &\h  0.4   &      37  &    1.2   &      33 &     2K   \newrow
\rm\sc Disp\param{5,3}    &     9 &     2K   &    5.4   &    11K   &    1.3   &      595 &    1.0   &        1 &\h  1.0   &      37  &   10.8   &      33 &    13K   \newrow
\rm\sc Disp\param{5,4}    &    10 &    15K   &   58.5   &   105K   &   16.4   &     6K   &   10.3   &      213 &\h 10.3   &      87  &  109     &      33 &   115K   \newrow
\rm\sc Disp\param{5,5}    &    11 &   151K   &       TO &        - &  476     &    53K   &  280     &     2K   &\h257     &     729  &       TO &      33 &        - \newrow
\rm\sc Disp\param{5,6}    &    12 &        ? &       TO &        - &       TO &        - &       TO &        - &       TO &    1131  &       TO &      33 &        - \newrow
\cmidrules

\rm\sc Mpat\param 4       &     9 &      384 &    0.5   &        0 &      N/A &          &      N/A &          &\h  0.5   &      37  &    0.6   &      33 &        0 \newrow
\rm\sc Mpat\param 5       &    11 &     4K   &    2.4   &        0 &      N/A &          &      N/A &          &    2.7   &      37  &\h  1.8   &      33 &        0 \newrow
\rm\sc Mpat\param 6       &    13 &    46K   &   50.6   &        0 &      N/A &          &      N/A &          &   73.2   &     214  &\h 21.5   &      33 &        0 \newrow
\rm\sc Mpat\param 7       &    15 &     645K &       TO &        - &       TO &        - &       TO &        - &       TO &     660  &\h359     &      33 &        0 \newrow
\rm\sc Mpat\param 8       &    17 &        ? &       TO &        - &       TO &        - &       TO &        - &       TO &     689  &       TO &      33 &        - \newrow
\cmidrules

\rm\sc MPC\param{2,5}     &     8 &       60 &    0.6   &      560 &    0.4   &        0 &          &          &\h  0.4   &      38  &    2.0   &      34 &     3K   \newrow
\rm\sc MPC\param{3,5}     &     9 &     3K   &   26.5   &    50K   &    3.0   &     3K   &    1.7   &        0 &\h  1.7   &      38  &   70.7   &      34 &    90K   \newrow
\rm\sc MPC\param{4,5}     &    10 &   314K   &       TO &        - &       TO &        - &  391     &    30K   &\h296     &     239  &       TO &      33 &        - \newrow
\rm\sc MPC\param{5,5}     &    11 &        ? &       TO &        - &       TO &        - &       TO &        - &       TO &     834  &       TO &      34 &        - \newrow
\cmidrules

\rm\sc Pi\param{5}        &     6 &      120 &\h  0.4   &        0 &      N/A &          &      N/A &          &    0.5   &      39  &   19.6   &      35 &        0 \newrow
\rm\sc Pi\param{6}        &     7 &      720 &\h  0.7   &        0 &      N/A &          &      N/A &          &    0.7   &      39  &  123     &      35 &        0 \newrow
\rm\sc Pi\param{7}        &     8 &     5K   &\h  3.5   &        0 &      N/A &          &      N/A &          &    4.0   &      45  &       TO &      34 &        - \newrow
\rm\sc Pi\param{8}        &     9 &    40K   &   48.1   &        0 &      N/A &          &      N/A &          &\h 42.9   &     246  &       TO &      34 &        - \newrow
\cmidrules

\rm\sc Pol\param{7,3}    &    14 &     3K   &   48.5   &    72K   &    2.9   &     1K   &\h  1.9   &        6 &    1.9   &      39  &   74.1   &      33 &    90K   \newrow
\rm\sc Pol\param{8,3}    &    15 &     4K   &  153     &   214K   &    5.5   &     3K   &\h  3.0   &       10 &    3.0   &      52  &  251     &      33 &   274K   \newrow
\rm\sc Pol\param{9,3}    &    16 &     5K   &  464     &   592K   &    9.5   &     5K   &\h  4.8   &       15 &    4.8   &      73  &       TO &      33 &        - \newrow
\rm\sc Pol\param{10,3}   &    17 &     7K   &       TO &        - &   17.2   &     9K   &\h  6.8   &       21 &    7.1   &      99  &       TO &      33 &        - \newrow
\rm\sc Pol\param{11,3}   &    18 &    10K   &       TO &        - &   27.2   &    12K   &\h  9.7   &       28 &   10.6   &     138  &       TO &      33 &        - \newrow
\rm\sc Pol\param{12,3}   &    19 &    12K   &       TO &        - &   46.3   &    20K   &\h 13.5   &       36 &   16.4   &     184  &       TO &      33 &        - \newrow

\bottomrule
\end{tabular}
}
\caption{\footnotesize Comparing QPOR and SDPOR.
Machine: Linux, Intel Xeon 2.4GHz.
TO: timeout after 8 min.
Columns are:
Th: \nr of threads;
Confs: maximal configurations;
Time in seconds,
Memory in MB;
SSB: Sleep-set blocked executions.
N/A: analysis with lower~$k$ yielded 0~SSBs.\vspace*{-3mm}}

\label{tab:sdpor}
\vspace{-05pt}
\end{table*}

\let\newrow\undefined
\let\param\undefined
\let\h\undefined
\let\cmidrules\undefined

In \Cref{tab:sdpor}, we present a comparison between \dpu and \nidhugg~\cite{AAAJLS15}, an efficient 
implementation of SDPOR for multithreaded C programs.
We run $k$-partial alternatives with $k \in \set{1, 2, 3}$ and optimal 
alternatives.
The number of SSB executions dramatically decreases as~$k$ increases.
With~$k=3$ almost no instance produces SSBs (except \textsc{MPC(4,5)}) and 
optimality is achieved with $k=4$. 
Programs with simple synchronization patterns, \eg, the \textsc{Pi} benchmark, 
are explored optimally both with $k=1$ and by SDPOR, while more complex 
synchronization patterns require $k > 1$.

Overall, if the benchmark exhibits many SSBs, the run time reduces as $k$ 
increases, and optimal exploration is the fastest option.
However, when the benchmark contains few SSBs (\cf, \textsc{Mpat}, \textsc{Pi},
\textsc{Poke}), $k$-partial alternatives can be slightly faster than optimal
POR, an observation inline with previous literature~\cite{AAJS14}.
Code profiling revealed that
when the comb is large and contains many solutions, both optimal
and non-optimal POR will easily find them, but optimal POR spends additional
time constructing a larger comb.
This suggests that optimal POR would profit from a lazy comb construction
algorithm.

\dpu is faster than~\nidhugg in the majority of the benchmarks because it can 
greatly reduce the number of SSBs.
In the cases where both tools explore the same set of executions,~\dpu is in 
general faster than~\nidhugg because it JIT-compiles the program, while~\nidhugg 
interprets it.
All the benchmark in~\cref{tab:sdpor} are data-race free, but
\nidhugg cannot be instructed to ignore data-races and will attempt to
revert them. \dpu was run with data-race detection disabled. Enabling it will
incur in approximatively 10\% overhead.
In contrast with previous observations~\cite{AAJS14,AAAJLS15}, the results 
in~\cref{tab:sdpor} show that SSBs can dramatically slow down the execution of 
SDPOR.

\subsection{Evaluation of the Tree-based Algorithms}

We now evaluate the efficiency of our tree-based algorithms from~\cref{sec:algo} answering:
(a) What are the average/maximal depths of the thread/lock sequential trees?
(b) What is the average depth difference on causality/conflict queries?
(c) What is the best step for branch skip lists?
We do not compare our algorithms against others
because to the best of our knowledge none is available (other than a naive
implementation of the mathematical definition of causality/conflict).

\begin{figure*}[ht]
	\subcaptionbox{Average depth of the tree nodes}{
		\includegraphics[height=3.5cm, width=0.49\textwidth]{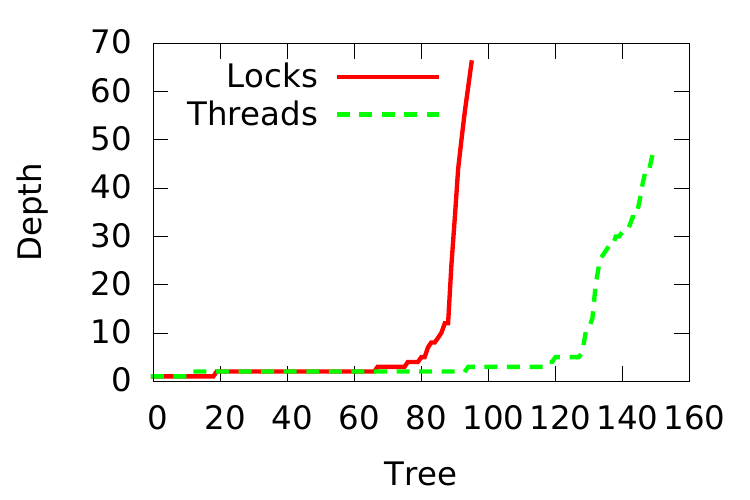}
		\label{tab:depth-avg}
	}
	\subcaptionbox{Maximum depth of the trees}{
		\includegraphics[height=3.5cm, width=0.485\textwidth]{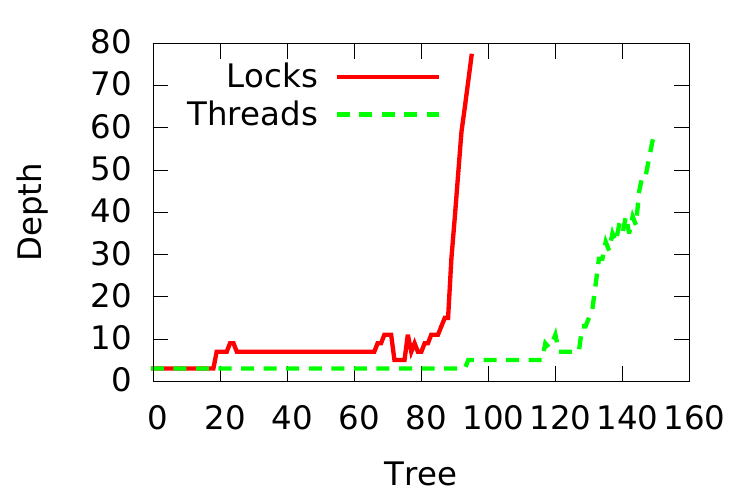}
		\label{tab:depth-max}
	}
	\subcaptionbox{\centering Depth-distance frequency on causality queries}
		{\includegraphics[width=.48\textwidth]{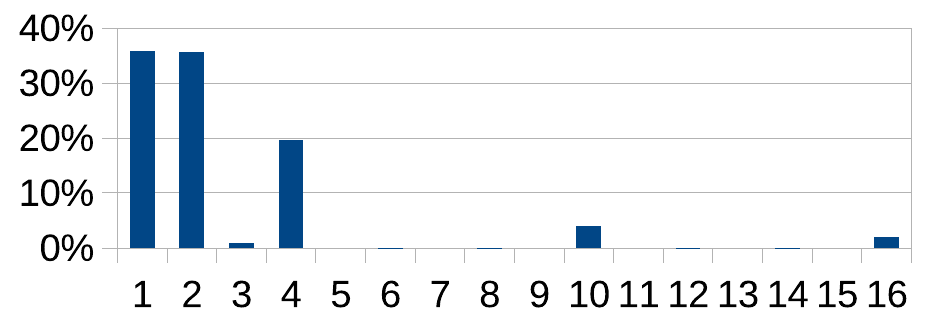}
		\label{tab:hist-cau}
	}
	\subcaptionbox{\centering Depth-distance frequency on conflict queries}
	{
		\includegraphics[width=.5\textwidth]{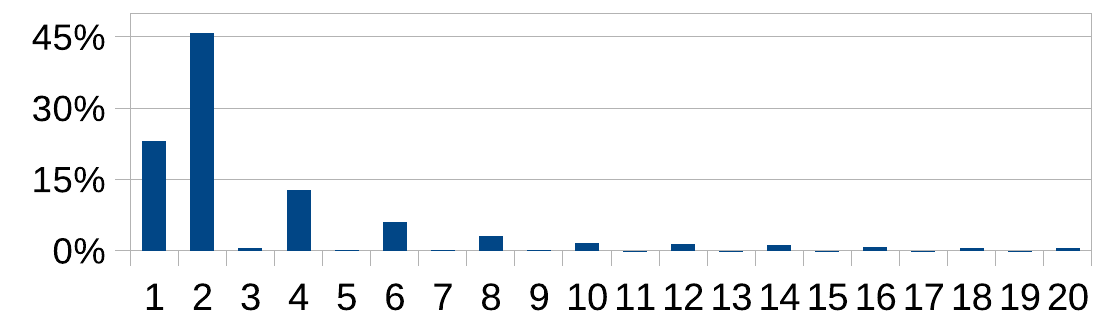}
		\label{tab:hist-cfl}
	}
        \vspace*{-3mm}
	\caption{(a), (b): depths of trees; (c), (d): frequency of depth distances\vspace*{-6mm}}
	\label{fig:depths}
\end{figure*}

We run~\dpu with an optimal exploration over 15 selected programs 
from~\cref{tab:sdpor}, with 380 to 204K maximal configurations in the unfolding.
In total, the 15 unfoldings contain 246 trees
(150 thread trees and 96 lock trees) with~5.2M nodes.
\cref{fig:depths} shows the average depth of the nodes in each tree
(subfigure a) and the maximum depth of the trees (subfigure b),
for each of the 246 trees.

While the average depth of a node is~22.7, as much as 80\% of the trees have a 
maximum depth of less than~8 nodes, and 90\% of them less than 16 nodes.
The average of 22.7 is however larger because deeper trees contain  
proportionally more nodes.
The depth of the deepest node of every tree was between~3 and~77.

We next evaluate depth differences in the causality and conflict queries over 
these trees.
\cref{fig:depths}~(a) and~(b) respectively show
the frequency of various depth distances associated to causality and 
conflict queries made by optimal POR.

Surprisingly, depth differences are very small for both causality and conflict
queries.
When deciding causality between events, as much as 92\% of the queries were for
tree nodes separated by a distance between~1 and~4, and 70\% had a difference 
of 1 or~2 nodes.
This means that optimal POR, and specifically the procedure that adds $\ex{C}$
to the unfolding (which is the main source of causality queries),
systematically performs causality queries which are trivial with the proposed 
data structures.
The situation is similar for checking conflicts: 82\% of the queries are about 
tree nodes whose depth difference is between~1 and~4.

\newcommand\newrow{\\[-1.0pt]}
\newcommand\param[1]{\scriptsize(#1)}

\newcommand\ocmidrules{
	\cmidrule(r){1-3}
	\cmidrule(r){4-6}
	\cmidrule(r){7-9}
}

\begin{wraptable}[24]{R}{0.50\textwidth}
\scriptsize
\setlength\tabcolsep{2.5pt}
\def\sep{\hspace{8pt}}
\def\tinysep{\hspace{5pt}}
\centering
\tt
\begin{tabular}[t]{l@{\hspace{-2pt}}rr@{\sep}rrr@{\tinysep}rrr}
	\toprule
	  \multicolumn{3}{l}{\rm Benchmark}
	& \multicolumn{3}{l}{\rm \dpu }  
	& \multicolumn{3}{l}{\rm \maple}

	\newrow
	\ocmidrules

	  {\rm Name}
	& {\rm LOC}
	& {\rm Th}

	& {\rm Time}
	& {\rm Ex}
	& {\rm R}
	
	& {\rm Time}
	& {\rm Ex}
	& {\rm R}

	\newrow
	\midrule

\rm\sc Add\param{2}    & 40K &  3 &   24.3 &     2 &   U  &    2.7 &    2 &  S    \newrow
\rm\sc Add\param{4}    & 40K &  5 &   25.5 &    24 &   U  &   34.5 &   24 &  U    \newrow
\rm\sc Add\param{6}    & 40K &  7 &   48.1 &   720 &   U  &     TO &  316 &  U    \newrow
\rm\sc Add\param{8}    & 40K &  9 &     TO &   14K &   U  &     TO &  329 &  U    \newrow
\rm\sc Add\param{10}   & 40K & 11 &     TO &   14K &   U  &     TO &  295 &  U    \newrow
\ocmidrules                                                                                                                                                        

\rm\sc Blk\param{5}    &  2K & 2 &    0.9 &      1 &  S  &    4.6 &    1 &  S    \newrow
\rm\sc Blk\param{15}   &  2K & 2 &    0.9 &      5 &  S  &   23.3 &    5 &  S    \newrow
\rm\sc Blk\param{18}   &  2K & 2 &    1.0 &    180 &  S  &     TO &  105 &  S    \newrow
\rm\sc Blk\param{20}   &  2K & 2 &    1.5 &   1147 &  S  &     TO &  106 &  S    \newrow
\rm\sc Blk\param{22}   &  2K & 2 &    2.6 &   5424 &  S  &     TO &  108 &  S    \newrow
\rm\sc Blk\param{24}   &  2K & 2 &   10.0 &    20K &  S  &     TO &  105 &  S    \newrow
\ocmidrules                                                                        
                                                                                                                                                 
\rm\sc Dnd\param{2,4}    & 16K & 3 &  11.1 &     80 &   U  &    122 &   80 &  U    \newrow
\rm\sc Dnd\param{4,2}    & 16K & 5 &  11.8 &     96 &   S  &    151 &   96 &  S    \newrow
\rm\sc Dnd\param{4,4}    & 16K & 5 &    TO &    13K &   U  &     TO &  360 &  U    \newrow
\rm\sc Dnd\param{6,2}    & 16K & 7 & 149.3 &   4320 &   S  &     TO &  388 &  S    \newrow
\ocmidrules                                                                        
                                                                                  
\rm\sc Mdl\param{1,4}    & 38K & 7 &  26.1 &      1 &   U  &    1.4 &    1 &  U    \newrow
\rm\sc Mdl\param{2,2}    & 38K & 5 &  29.2 &      9 &   U  &   13.3 &    9 &  U    \newrow
\rm\sc Mdl\param{2,3}    & 38K & 5 &  46.2 &    576 &   U  &     TO &  304 &  U    \newrow
\rm\sc Mdl\param{3,2}    & 38K & 7 &  31.1 &    256 &   U  &    402 &  256 &  U    \newrow
\rm\sc Mdl\param{4,3}    & 38K & 9 &    TO &    14K &   U  &    TO &   329 &  U    \newrow
\ocmidrules                                                                      
                                                                                
\rm\sc Pla\param{1,5}    & 41K & 2 &  22.8 &      1 &   U  &    1.7 &    1 &  U    \newrow
\rm\sc Pla\param{2,4}    & 41K & 3 &  37.2 &     80 &   U  &  142.4 &   80 &  U    \newrow
\rm\sc Pla\param{4,3}    & 41K & 5 & 160.5 &   1368 &   U  &     TO &  266 &  U    \newrow
\rm\sc Pla\param{6,3}    & 41K & 7 &    TO &   4580 &   U  &     TO &  269 &  U    \newrow

	\bottomrule
	
\end{tabular}

\caption{\footnotesize Comparing DPU with Maple (same machine).
LOC: lines of code;
Execs: \nr of executions;
R: safe or unsafe.
Other columns as before.
Timeout: 8 min.}
\label{tab:maple}
\end{wraptable}

\let\newrow\undefined
\let\param\undefined
\let\cmidrules\undefined

These experiments show that most queries on the causality trees require 
very short walks, which strongly drives to use the data structure proposed 
in~\cref{sec:algo}.
Finally, we chose a (rather arbitrary) skip step of 4.
We observed that other values do not significantly impact the
run time/memory consumption for most benchmarks, since the depth difference on 
causality/conflict requests is very low.

\subsection{Evaluation Against the State-of-the-art on System Code}

We now evaluate the scalability and applicability of~\dpu on five multithreaded 
programs in two Debian packages:
~\emph{blktrace}~\cite{BLKT}, a block layer I/O tracing mechanism, and 
~\emph{mafft}~\cite{MAFFT}, a tool for multiple alignment of amino acid or nucleotide sequences.
The code size of these utilities ranges from 2K to 40K LOC, and
\emph{mafft} is parametric in the number of threads.

We compared~\dpu against \maple \cite{YNPP12}, a state-of-the-art testing tool 
for multithreaded programs, as the top ranked verification tools from
SVCOMP'17 are still unable to cope with such large and complex multithreaded 
code.
Unfortunately we could not compare against \nidhugg because
it cannot deal with the (abundant) C-library calls in these programs.

\Cref{tab:maple} presents our experimental results.
We use~\dpu with optimal exploration and the modified version of~\maple
used in~\cite{TDB16}.
To test the effectiveness of both approaches in state space coverage 
and bug finding, we introduce bugs in 4 of the benchmarks 
(\textsc{Add,Dnd,Mdl,pla}).
For the safe benchmark~\textsc{Blk}, we perform exhaustive state-space
exploration using \maple's DFS mode.
On this benchmark, \dpu outperfors~\maple by several orders of 
magnitude:~\dpu explores up to 20K executions covering the entire
state space in 10s, while~\maple only explores up to 108 executions in
8~min.  

For the remaining benchmarks, we use the random scheduler of~\maple, considered
to be the best baseline for bug finding~\cite{TDB16}.
First, we run~\dpu to retrieve a bound on the number of random executions to
answer whether both tools are able to find the bug within the same number 
of executions.
\maple found bugs in all buggy programs (except for one variant
in~\textsc{Add}) even though~\dpu greatly outperforms and is able to 
achieve much more state space coverage.

\subsection{Profiling a Stateless POR}


In order
to understand the cost of each component of the algorithm,
we profile~\dpu on a selection of 7~programs from~\cref{tab:sdpor}.
\dpu spends between 30\% and 90\% of the run time executing the program (65\% in average).
The remaining time is spent computing alternatives, distributed as follows:
adding events to the event structure (15\% to 30\%),
building the spikes of a new comb (1\% to 50\%),
searching for solutions in the comb (less than~5\%), and
computing conflicting extensions (less than~5\%).
Counterintuitively, building the \emph{comb} is more expensive than exploring it, even 
in the optimal case.
Filling the spikes 
seems to be more memory-intensive than exploring the comb, which exploits data locality.
%

\section{Conclusion} 
\label{sec:concl}

We have shown that computing alternatives in an optimal DPOR exploration is NP-complete.
To mitigate this problem, we introduced a new approach to compute alternatives in polynomial time, approximating the optimal exploration with a user-defined constant.
Experiments conducted on benchmarks including Debian packages show that our 
implementation outperforms current verification tools and uses appropriate data structures.
Our profiling results show that running the program is often more expensive 
than computing alternatives.
Hence, efforts in reducing the number of redundant executions, even if
significantly costly, are likely to reduce the overall execution time.

\bibliographystyle{splncs}
\bibliography{main.bbl}

\newpage
\appendix
\section{Additional Basic Definitions}
\label{sec:basic-defs}

In this section we introduce a number of definitions that were excluded from the
body of the paper owing to space constraints.

\paragraph{Labelled Transition Systems.}

We defined an \lts\ semantics for programs in \cref{sec:prelim} without first
providing a general definition of \lts{}s.
An \lts~\cite{CGP99} is a structure
$M \eqdef \tup{\Sigma, \to, A, s_0}$, where
$\Sigma$ are the \emph{states},
$A$ the \emph{actions},
${\to} \subseteq \Sigma \times A \times \Sigma$ the transition relation, and
$s_0 \in \Sigma$ an \emph{initial state}.
If $s \fire{a} s'$ is a transition, 
the action~$a$ is \emph{enabled} at~$s$
and~$a$ can \emph{fire} at~$s$ to produce~$s'$.
We let~$\enabl s$ denote the set of actions enabled at~$s$.

A sequence $\sigma \eqdef a_1 \ldots a_n \in A^*$ is a \emph{run} when
there are states $s_1, \ldots, s_n$ satisfying
$s_0 \fire{a_1} s_1 \ldots \fire{a_n} s_n$.
We define $\statee \sigma \eqdef s_n$.
We let $\runs M$ denote the set of all runs of~$M$,
and $\reach M \eqdef \set{\statee \sigma \in \Sigma \colon \sigma \in \runs M}$
the set of all \emph{reachable states}.

\paragraph{Prime Event Structures.}

Let $\les \eqdef \tup{E, {<}, {\cfl}}$ be a \pes.
Two events $e, e' \in E$ are in \emph{immediate conflict}
if $e \cfl e'$ but both $\causes e \cup [e]$ and $[e] \cup \causes{e'}$ are free
of conflict.
Given a set $U \subseteq E$, we denote by $\ficfl[U] e$ the set of
events in~$U$ that are in immediate conflict with~$e$.

\paragraph{Unfolding semantics of an \lts.}

In \cref{sec:prelim} we defined the unfolding semantics of a program
(\cref{def:unfd}).
Now we give a slightly more general definition for \lts{}s instead of programs.
The definitions are almost identical, the only differences are found in the
first three lines of the definition.
In particular the four fixpoint rules are exactly the same.
The reason why we give now this definition over \lts{} is because we will use it
to define unfolding semantics for Petri nets in the proof of~\cref{thm:cex-np}.

\begin{definition}[Unfolding of an \lts~\cite{RSSK15}]
\label[definition]{def:unf-lts}
Given an \lts~$M \eqdef \tup{\Sigma, \to, A, s_0}$ and some independence
relation ${\indep} \subseteq A \times A$ on $M$, the
\emph{unfolding of~$M$ under~$\indep$}, denoted $\unf{M,\indep}$,
is the \pes over~$A$ constructed by the following fixpoint rules:
\begin{enumerate}[topsep=0pt]
\item
  Start with a \pes $\les \eqdef \tup{E, <, {\cfl}, h}$
  equal to $\tup{\emptyset, \emptyset, \emptyset, \emptyset}$.
\item
  Add a new event $e \eqdef \tup{a,C}$ to~$E$ for any
  configuration $C \in \conf \les$ and any action $a \in A$ such that
  $a$~is enabled at $\statee C$ and
  $\lnot (a \indep h(e'))$ holds for every $<$-maximal event $e'$ in~$C$.
\item
  For any new $e$ in $E$, update $<$, $\cfl$, and $h$ as follows:
  \begin{itemize}
  \item
    for every $e' \in C$, set $e' < e$;
  \item
    for any $e' \in E \setminus C$,
    set $e' \cfl e$
    if $e \ne e'$ and $\lnot (a \indep h(e'))$;
  \item
    set $h(e) \eqdef a$.
  \end{itemize}
\item
  Repeat steps 2 and 3 until no new event can be added to~$E$;
  return $\les$.
\end{enumerate}
\end{definition}

Obviously, both \cref{def:unfd} and \cref{def:unf-lts} produce the same
unfolding when applied to a program.
That is, for any program~$P$ and independence~$\indep$
on~$M_P$, we have that
$\unf{P,\indep}$ (\cref{def:unfd}) is equal to
$\unf{M_P,\indep}$ (\cref{def:unf-lts}).

\paragraph{Petri nets.}

A Petri net~\cite{Mur89} is a model of a concurrent system.
Formally,
a \emph{net} is a tuple $N \eqdef \tup{P, T, F, m_0}$,
where~$P$ and~$T$ are disjoint finite sets of \emph{places}
and \emph{transitions},
$F \subseteq (P \times T) \cup (T \times P)$
is the \emph{flow relation},
and $m_0 \colon P \to \nat$ is the \emph{initial marking}.
$N$ is called \emph{finite} if $P$ and $T$ are finite.
Places and transitions together are called \emph{nodes}.

For $x \in P \cup T$, let $\pre x \eqdef \{y \in P \cup T
\colon (y,x) \in F\}$ be the \emph{preset},
and $\post x \eqdef \{y \in P \cup T \colon (x,y) \in F\}$  the
\emph{postset} of $x$.
The state of a net is represented by a marking.
A \emph{marking} of~$N$ is a function $m \colon P \to \nat$ that
assigns \emph{tokens} to every place.
A transition~$t$ is \emph{enabled} at a marking~$m$ iff for any $p \in \pre t$
we have $m(p) \geq 1$.

We give semantics to nets using transition systems.
We associate~$N$ with a transition system
$M_N \eqdef \tup{\Sigma, {\to}, A, m_0}$ where
$\Sigma \eqdef P \to \nat$ is the set of markings,
$A \eqdef T$ is the set of transitions, and
${\to} \subseteq \Sigma \times A \times \Sigma$
contains a triple $m \fire{t} m'$ exactly
when, for any $p \in \pre t$ we have~$m(p) \ge 1$,
and for any $p \in P$ we have
$m'(p) = m(p) - |\set p \cap \pre t| + |\set p \cap \post t|$.
We call~$N$ $k$-safe when for any reachable marking $m \in \reach{M_N}$ we
have $m(p) \le k$, for $p \in P$.

\section{General Lemmas}
\label{sec:general-lemmas}

\newcommand\calls{\mathrel{\triangleright}}
\newcommand\callsl{\mathrel{\triangleright_l}}
\newcommand\callsr{\mathrel{\triangleright_r}}

For the rest of this section,
we fix an \lts~$M \eqdef \tup{\Sigma, A, {\to}, s_0}$
and an independence relation $\indep$ on~$M$.
We assume that $\runs M$ is a finite set of finite sequences.
Let $\unf{M,\indep} \eqdef \tup{E, <, {\cfl}, h}$ be the unfolding of~$M$
under~$\indep$, which we will abbreviate as~$\uunf$.
Note that $\uunf$ is finite because of our assumption about~$\runs M$.
We assume that $\uunf$ is the input \pes provided to $\cref{a:a1}$.
Finally, without loss of generality we assume that $\uunf$ contains a special
event~$\bot$ that is a causal predecessor of any other event in~$\uunf$.

\Cref{a:a1} is recursive, each call to \explore{$C,D,A$} 
yields either no recursive call, if the function returns at
\cref{l:ret},
or one single recursive call (\cref{l:left}),
or two (\cref{l:left} and \cref{l:right}).
Furthermore, it is non-deterministic, as $e$ is chosen from either the set
$\en C \setminus D$ or the set $A \cap \en C$,
which in general are not singletons.
As a result, the configurations explored by it
may differ from one execution to the next.

For each run of the algorithm on~$\uunf$
we define the \emph{call graph} explored by \cref{a:a1}
on that run as a directed graph $\tup{B, {\calls}}$
representing the actual exploration of~$\uunf$.
Different executions will in general yield different call graphs.

The nodes~$B$ of the call graph are
4-tuples of the form $\tup{C,D,A,e}$, where $C,D,A$ are the parameters of
a recursive call made to the funtion \explore{$\cdot,\cdot,\cdot$}, and
$e$ is the event selected by the algorithm immediately before \cref{l:left}.
More formally,~$B$ contains exactly all tuples $\tup{C,D,A,e}$ satisfying
that
\begin{itemize}
\item
   $C$, $D$, and $A$ are sets of events of the unfolding $\uunf$;
\item
   during the execution of \explore{$\emptyset,\emptyset,\emptyset$},
   the function \explore{$\cdot,\cdot,\cdot$} has been recursively called
   with $C,D,A$ as, respectively, first, second, and third argument;
\item
   $e \in E$ is the event selected by \explore{$C,D,A$} immediately
   before \cref{l:left} if $\en C \not\subseteq D$.
   When $\en C \subseteq D$ we define $e \eqdef \bot$.
   \footnote{Observe that in this case, if $\en C \subseteq D$, the execution of
   \explore{$C,D,A$} never reaches \cref{l:left}.}
\end{itemize} 
The edge relation of the call graph,
${\calls} \subseteq B \times B$, represents the recursive calls made by
\explore{$\cdot,\cdot,\cdot$}.
Formally, it is the union of two disjoint relations
${\calls} \eqdef {\callsl} \uplus {\callsr}$, defined as follows.
We define that
\[
\tup{C,D,A,e} \callsl \tup{C',D',A',e'}
\text{ ~ and that ~ }
\tup{C,D,A,e} \callsr \tup{C'',D'',A'',e''}
\]
iff the execution of \explore{$C,D,A$} issues a recursive
call to, respectively,
\explore{$C',D',A'$} at \cref{l:left} and
\explore{$C'',D'',A''$} at \cref{l:right}.
Observe that $C'$ and $C''$ will necessarily be different
(as $C' = C \cup \set e$, where $e \notin C$, and $C'' = C$), and therefore the
two relations are disjoint sets.
We distinguish the node
\[
b_0 \eqdef \tup{\emptyset,\emptyset,\emptyset,\bot}
\]
as the \emph{initial node}, also called the \emph{root node}.
Observe that $\tup{B, {\calls}}$ is by definition a weakly connected
digraph, as there is a path from the node $b_0$ to every other node in~$B$.
We refer to
$\callsl$ as the \textit{left-child} relation
and $\callsr$ as the \textit{right child} relation.

\begin{lemma}
\label{l:general}
Let $\tup{C,D,A,e} \in B$ be a state of the call graph. We have that
\begin{itemize}
\item
   $D \cap A = \emptyset$;
   \eqtag{e:basic0}
\item
   event $e$ is such that $e \in \en C \setminus D$;
   \eqtag{e:basic1}
\item
   $C$ is a configuration;
   \eqtag{e:basic2}
\item
   $C \cup A$ is a configuration and $C \cap A = \emptyset$;
   \eqtag{e:basic3}
\item
   $D \subseteq \ex C$;
   \eqtag{e:basic4}
\end{itemize}
\end{lemma}
\begin{proof}
   Proving~\cref{e:basic1} is immediate, assuming that~\cref{e:basic0} holds.
   In \cref{a:a1},
   observe both branches of the conditional statement where~$e$ is selected.
   If~$e$ is slected by the \textit{then} branch,
   clearly~$e \in \en C \setminus D$.
   If~$e$ is selected by the \textit{else} branch,
   clearly $e \in \en C$. But, by~\cref{e:basic0} $e \notin D$, as~$e \in A$
   and~$A$ is disjoint with~$D$. Therefore $e \in \en C \setminus D$.
   In both cases~\cref{e:basic1} holds, what we wanted to prove.

   All remaining items,
   \cref{e:basic0,e:basic2,e:basic3,e:basic4}, will be
   shown by induction on the length $n \ge 0$ of any path
   \[
   b_0 \calls b_1 \calls \ldots \calls b_{n-1} \calls b_n
   \]
   on the call graph, starting from the initial node
   and leading to 
   $b_n \eqdef \tup{C, D, A, e}$
   For $i \in \set{0, \ldots, n}$ we define
   $\tup{C_i,D_i,A_i,e_i} \eqdef b_i$.

   We start showing \cref{e:basic0}.
   \emph{Base case.}
   $n = 0$ and $D = A = \emptyset$. The result holds.
   \emph{Step.}
   Assume that $D_{n-1} \cap A_{n-1} = \emptyset$ holds.
   We have
   \[
      D \cap A = D_n \cap A_n = D_{n-1} \cap (A_{n-1} \setminus \set e) =
      D_{n-1} \cap A_{n-1} = \emptyset
   \]
   because removing event~$e$ from~$A$ will not increase the number of events
   shared by~$A$ and~$D$.

   We now show \cref{e:basic2}, also by induction on~$n$.
   \emph{Base case.}
   $n = 0$ and $C = \emptyset$. The set $\emptyset$ is a configuration.
   \emph{Step.}
   Assume $C_{n-1}$ is a configuration. If
   $b_{n-1} \callsl b_n$,
   then $C = C_{n-1} \cup \set e$ for some event $e \in \en C$, as stated
   in \cref{e:basic1}. By definition,
   $C$ is a configuration.
   If 
   $b_{n-1} \callsr b_n$,
   then $C = C_{n-1}$. In any case $C$ is a configuration.

   We show \cref{e:basic3}, by induction on $n$.
   \emph{Base case.}
   $n = 0$. Then $C = \emptyset$ and $A = \emptyset$. Clearly $C \cup A$ is a
   configuration and $C \cup A = \emptyset$.
   \emph{Step.}
   Assume that $C_{n-1} \cup A_{n-1}$ is a configuration and that
   $C_{n-1} \cap A_{n-1} = \emptyset$.
   We have two cases.
   \begin{itemize}
   \item
      Assume that $b_{n-1} \callsl b_n$.
      If $A_{n-1}$ is empty, then $A$ is empty as well. Clearly
      $C \cup A$ is a configuration and $C \cap A$ is empty.
      If $A_{n-1}$ is not empty, then
      $C = C_{n-1} \cup \set e$ and
      $A = A_{n-1} \setminus \set e$, for some
      $e \in A_{n-1}$, and we have
      \[
         C \cup A =
         (C_{n-1} \cup \set e ) \cup (A_{n-1} \setminus \set e ) =
         C_{n-1} \cup A_{n-1},
      \]
      so $C \cup A$ is a configuration as well.
      We also have that $C \cap A = C_{n-1} \cap A_{n-1}$ (recall that
      $e \notin C$), so $C \cap A$ is empty.
   \item
      Assume that $b_{n-1} \callsr b_n$ holds.
      Then we have $C = C_{n-1}$ and also
      $A = J \setminus C_{n-1}$ for some
      $J \in \alternatives{$C_{n-1},D \cup \set e$}$.
      Since~$J$ is a clue,
      from \cref{def:alt-fun,def:clue},
      we know that $C_{n-1} \cup J$ is a configuration.
      As a result,
      \[
         C \cup A =
         C_{n-1} \cup (J \setminus C_{n-1}) =
         C_{n-1} \cup J,
      \]
      and therefore $C \cup A$ is a configuration.
      Finally, by construction of $A$ at \cref{l:right},
      we clearly have $C \cap A = \emptyset$.
   \end{itemize}

   We show \cref{e:basic4}, again, by induction on $n$.
   \emph{Base case}.
   $n = 0$ and $D = \emptyset$. Then \cref{e:basic4} clearly holds.
   \emph{Step.}
   Assume that \cref{e:basic4} holds for $\tup{C_i,D_i,A_i,e_i}$, with
   $i \in \set{0, \ldots, n - 1}$. We show that it holds for $b_n$.
   As before, we have two cases.
   \begin{itemize}
   \item
      Assume that $b_{n-1} \callsl b_n$.
      We have that
      $D = D_{n-1}$ and that
      $C = C_{n-1} \cup \set{e_{n-1}}$.
      We need to show that for all $e' \in D$
      we have $\causes{e'} \subseteq C$ and $e' \notin C$.
      By induction hypothesis we know that
      $D = D_{n-1} \subseteq \ex{C_{n-1}}$, so clearly
      $\causes{e'} \subseteq C_{n-1} \subseteq C$.
      We also have that $e' \notin C_{n-1}$, so we only need to check that
      $e' \ne e_{n-1}$.
      By~\cref{e:basic1} applied to~$b_{n-1}$ we have that
      $e_{n-1} \notin D_{n-1} = D$. That means that $e' \ne e_{n-1}$.
   \item
      Assume that $b_{n-1} \callsr b_n$.
      We have that $D = D_{n-1} \cup \set{e_{n-1}}$, and by hypothesis we
      know that $D_{n-1} \subseteq \ex{C_{n-1}} = \ex C$.
      As for $e_{n-1}$, by \cref{e:basic1} we know that
      $e_{n-1} \in \en{C_{n-1}} = \en C \subseteq \ex C$.
      As a result, $D \subseteq \ex C$.
   \end{itemize}
\end{proof}

\begin{lemma}
   Let $b \eqdef \tup{C,D,A,e}$ and $b' \eqdef \tup{C',D',A',e'}$ be
   two nodes of the call graph such that $b \calls b'$.
   Then
   \begin{itemize}
      \item $C \subseteq C'$ and $D \subseteq D'$; \eqtag{e:step1}
      \item if $b \callsl b'$, then $C \subsetneq C'$; \eqtag{e:step2}
      \item if $b \callsr b'$, then $D \subsetneq D'$. \eqtag{e:step3}
   \end{itemize}
\end{lemma}
\begin{proof}
   If $b \callsl b'$, then
   $C' = C \cup \set e$ and
   $D' = D$.
   Then all the three statements hold.
   If $b \callsr b'$, then
   $C' = C$ and
   $D' = D \cup \set e$.
   Similarly, all the three statements hold.
\end{proof}

\begin{lemma}
   \label{l:twoconfs}
   If $C \subseteq C'$ are two finite configurations, then
   $\en C \cap (C' \setminus C) = \emptyset$
   iff
   $C' \setminus C = \emptyset$.
\end{lemma}
\begin{proof}
If there is some $e \in \en C \cap (C' \setminus C)$,
then $e \notin C$ and $e \in C'$, so $C' \setminus C$ is not empty.
If there is some $e' \in C' \setminus C$, then there is some $e''$ event
that is $<$-minimal in $C' \setminus C$.
As a result, $\causes{e''} \subseteq C$. Since $e'' \notin C$
and $C \cup \set{e''}$ is a configuration
(as $C\cup \set{e''} \subseteq C'$), we have that $e'' \in \en C$.
Then $\en C \cap (C' \setminus C)$ is not empty.
\end{proof}

\section{Termination Proofs}
\label{sec:termination}

\begin{lemma}
   \label{l:finitepath}
   Any path $b_0 \calls b_1 \calls b_2 \calls \ldots$
   in the call graph starting from $b_0$ is finite.
\end{lemma}
\begin{proof}
   By contradiction.
   Assume that $b_0 \calls b_1 \calls \ldots$ is an infinite path in the
   call graph.
   For $0 \le i$, let $\tup{C_i,D_i,A_i,e_i} \eqdef b_i$.
   Recall that $\uunf$ has finitely many events, finitely many finite
   configurations, and no infinite configuration.
   Now, observe that the number of times that $C_i$ and $C_{i+1}$ are
   related by $\callsl$ rather than $\callsr$ is finite, since every time
   \explore{$\cdot,\cdot,\cdot$} makes a recursive call at \cref{l:left} it
   adds one event to $C_i$, as stated by~\cref{e:step2}.
   More formally, the set
   \[ L \eqdef \set{i \in \N \colon C_i \callsl C_{i+1}} \]
   is finite.
   As a result it has a maximum, and its successor
   $k \eqdef 1 + \max_< L$ is an index in the path such that
   for all $i \ge k$ we have $C_i \callsr C_{i+1}$, \ie, the function only
   makes recursive calls at \cref{l:right}.
   We then have that $C_i = C_k$, for $i \ge k$, and by~\cref{e:basic4},
   that $D_i \subseteq \ex{C_k}$.
   Since~$\uunf$ is finite, note that~$\ex{C_k}$ is finite as well.
   But, as a result of \cref{e:step2} the sequence
   \[
   D_k \subsetneq D_{k+1} \subsetneq D_{k+2} \subsetneq \ldots
   \]
   is an infinite increasing sequence.
   This is a contradiction, as for sufficiently large $j \ge 0$ we will
   have that $D_{k+j}$ will be larger than $\ex{C_k}$, yet
   $D_{k+j} \subseteq \ex{C_k}$.
\end{proof}

\thmtermination*
\begin{proof}
   The statement of the theorem refers to~\cref{a:a1}, but we instead prove it
   for~\cref{a:a1}.
   Remark that \cref{a:a1} makes calls to two functions, namely,
   \remove{$\cdot$}
   and
   \alternatives{$\cdot, \cdot$}.
   Clearly both of them terminate
   (the loop in \remove{$\cdot$} iterates over a finite set).
   Since we gave no algorithm to compute \alternatives{$\cdot$},
   we will assume we employ one that terminates on every input.

   Now, observe that there is no loop in \cref{a:a1}.
   Thus any non-terminating execution of \cref{a:a1} must perform a
   non-terminating sequence of recursive calls, which entails the existence
   of an infinite path in the call graph associated to the execution.
   Since, by~\cref{l:finitepath}, no infinite path exist in the call graph,
   \cref{a:a1} always terminates.
\end{proof}

\section{Completeness Proofs}
\label{sec:completeness}

\begin{lemma}
   \label{l:compl.step}
   Let $b \eqdef \tup{C,D,A,e} \in B$ be a node in the call graph and
   $\hat C \subseteq E$ an arbitrary maximal configuration of $\uunf$
   such that $C \subseteq \hat C$ and $D \cap \hat C = \emptyset$.
   Then exactly one of the following statements holds:
   \begin{itemize}
   \item
     Either $C$ is a maximal configuration of $\uunf$, or
   \item
     $C$ is not maximal but $\en C \subseteq D$, or
   \item
     $e \in \hat C$ and~$b$ has a left child, or
   \item
     $e \notin \hat C$ and~$b$ has a right child.
   \end{itemize}
\end{lemma}
\begin{proof}
	If $C$ is maximal, then the first statement holds and~$b$ has no successor in
   the call graph, so none of the other three statements hold and we are done.

   So assume that $C$ is not maximal.
   Then $\en C \ne \emptyset$.
   Now, if $\en C \subseteq D$ holds then the second statement is true and none
   of the others is
   (as \cref{a:a1} does not make any recursive call in this case).

   So assume also that $\en C \not\subseteq D$.
   That implies that~$b$ has at least one left child.
   If $e \in \hat C$, then we are done, as the second statement holds and none
   of the others hold.

   So finally, assume that $e \notin \hat C$, we need to show that the third
   statement holds, \ie that~$b$ has right child.
   By \cref{def:alt-fun} we know that the set of clues returned by the call to
   \alternatives{$C, D \cup \set e$} will be non-empty, as
   there exists a maximal configuration $\hat C$ such that
   $C \subseteq \hat C$ (by hypothesis) and
   \[
      \hat C \cap (D \cup \set e) =
      (\hat C \cap D) \cup (\hat C \cap \set e) =
      \hat C \cap D =
      \emptyset.
   \]
   This means that \cref{a:a1} will make a recursive call at line
   \cref{l:right} and~$b$ will have a right child.
   This shows that the last statement holds.
   And clearly none of the other statements holds in this case.
\end{proof}

\begin{lemma}
   \label{l:compl.findit}
   For any node $b \eqdef \tup{C,D,\cdot,e} \in B$ in the call graph and any
   maximal configuration $\hat C \subseteq E$ of $\uunf$,
   if
   \[ C \subseteq \hat C \text{ and } D \cap \hat C = \emptyset, \]
   then there is a node~$b' \eqdef \tup{C',\cdot,\cdot,\cdot} \in B$ such that
   $b \calls^* b'$ and $\hat C = C'$.
\end{lemma}
\begin{proof}
   The proof works by explicitly constructing a path from~$b$ to~$b'$ using an
   iterated application of~\cref{l:compl.step}.
   
   Since $C \subseteq \hat C$ and $D \cap \hat C = \emptyset$, we can apply
   \cref{l:compl.step} to $b$ and $\hat C$ and conclude that
   exactly one of the four statements in that Lemma will be true at~$b$.
   If~$C$ is maximal, then necessarily $C = \hat C$ and we are done.
   If~$C$ is not maximal,
   then it must be the case that $\en C \not\subseteq D$
   and~$b$ has at least one left child.
   This is because by~\cref{l:twoconfs} we have that
   \[
      \en C \cap (\hat C \setminus C) = \emptyset
      \text{ iff }
      \hat C \setminus C = \emptyset.
   \]
   Since~$C$ is not maximal $\hat C \setminus C \ne \emptyset$ and we see that
   $\en C \cap \hat C$ must be non-empty.
   Now, since~$\hat C$ and~$D$ are disjoint, the event(s) in
   $\en C \cap \hat C$ are not in~$D$, and so
   $\en C$ contains events which are not contained in~$D$.

   Since $\en C \not\subseteq D$ we have that the second statement
   in~\cref{l:compl.step} does not hold, and so either the third or the
   fourth statement have to be hold.

   Now, $b$ has a left child and two cases are possible,
   either $e \in \hat C$ or not.
   If $e \in \hat C$ we
   let~$b_1 \eqdef \tup{C_1, D_1, \cdot, e_1}$ be the left child of~$b$,
   with $C_1 \eqdef C \cup \set e$ and $D_1 \eqdef D$.
   If $e \notin \hat C$, then only the last statement of \cref{l:compl.step} can
   hold and we know that~$b$ has a right child.
   Let~$b_1 \eqdef \tup{C_1, D_1, \cdot, e_1}$,
   with $C_1 \eqdef C$ and $D_1 \eqdef D \cup \set e$ be that child.
   Observe that in both cases $C_1 \subseteq \hat C$ and
   $D_1 \cap \hat C = \emptyset$.

   If $C_1$ is maximal, then necessarily $C_1 = \hat C$, we
   take~$b' \eqdef b_1$ and we have finished.
   If not, we can reapply \cref{l:compl.step} at~$b_1$ and
   make one more step into one of the children $b_2$ of~$b_1$.
   If $C_2$ is still not maximal (thus different from $\hat C$)
   we need to repeat the argument starting from~$b_2$ only a finite
   number~$n$ of times until we reach a
   node~$b_n \eqdef \tup{C_n, D_n, \cdot, \cdot}$ where $C_n$ is a maximal
   configuration.
   This is because every time we repeat the argument on a non-maximal
   node~$b_i$ we advance one step down in the call graph,
   and by~\cref{l:finitepath} all paths in the graph starting from the root are
   finite.
   So eventually we find a leaf node~$b_n$ where $C_n$ is maximal and satisfies
   $C_n \subseteq \hat C$.
   This implies that $C_n = \hat C$, and we can take $b' \eqdef b_n$.
\end{proof}

\thmcompleteness*
\begin{proof}
   We need to show that for every maximal
   configuration~$\hat C \subseteq E$ we can find a node
   $b \eqdef \tup{C, \cdot, \cdot, \cdot}$ in~$B$
   such that $\hat C = C$.
   This is a direct consequence of \cref{l:compl.findit}.
   Consider the root node of the tree,
   $b_0 \eqdef \tup{C_0,D_0,A_0,\bot}$, where
   $C_0 = D_0 = A_0 = \emptyset$.
   Clearly $C_0 \subseteq \hat C$ and
   $D_0 \cap \hat C = \emptyset$, so
   \cref{l:compl.findit} applies to $\hat C$ and $b_0$,
   and establishes the existence of the aforementioned node~$b$.
\end{proof}


\section{Complexity Proofs}
\label{sec:complexity-proofs}

\pesaltnp*
\begin{proof}
We first prove that the problem is in~NP. Let us non-deterministically choose a configuration
$J \subseteq E$. We then check that $J$ is an alternative to~$D$ after~$C$:
\begin{itemize}
  \item $J \cup C$ is a configuration can be checked in linear time: The first condition
  for $J \cup C$ to be a configuration is that $\forall e \in J \cup C: \causes e \subseteq
  J \cup C$. Since $J$ is a configuration, this condition holds for all $e \in J$.
  Similarly, as $C$ is a configuration, it also holds for all $e \in C$. The second
  condition is that $\forall e_1, e_2 \in J \cup C : \lnot (e_1 \cfl e_2)$. This is
  true
  for $e_1,e_2 \in J$ and $e_1,e_2 \in C$. If $e_1 \in J \land e_2 \in C$ (or the
  converse), we have to effectively check that $\lnot (e_1 \cfl e_2)$. Checking if
  two events $e_1$ and $e_2$ are in conflict is linear on the size
  of $[e_1] \cup [e_2]$.
  \item Every event $e_1 \in D$ must be in immediate conflict with an event $e_2 \in
  J$. Thus, there are at most $|D|\cdot|J|$ checks to perform, each in linear time
  on the size of $[e_1] \cup [e_2]$. Hence, this is in $O(n^2)$.
\end{itemize}

We now prove that the problem is NP-hard, by reduction from the 3-SAT problem.
Let $\set{v_1, \ldots, v_n}$ be a set of Boolean variables.
Let $\phi \eqdef c_1 \land \ldots \land c_m$ be a 3-SAT formula, where
each clause $c_i \eqdef l_i \lor l'_i \lor l''_i$ comprises three literals.
A literal is either a Boolean variable $v_i$ or its negation $\setneg{v_i}$.

Formula $\phi$ can be modelled by a PES $\les_\phi \eqdef \tup{E,{<},\cfl,h}$ constructed
as follows:
\begin{itemize}
  \item For each variable $v_i$ we create two events $t_i$ and $f_i$ in~$E$, and put
    them in immediate conflict, as they correspond to the satisfaction of $v_i$ and
    $\setneg{v_i}$, respectively.
  \item The set $D$ of events to disable contains one event $d_j$ per clause $c_j$.
  Such a $d_j$ has to be in immediate conflict with the events modelling the literals
  in clause $c_j$. Hence it is in conflict with 1, 2, or 3 $t$ or $f$ events.
  \item There is no causality: $< \eqdef\emptyset$.
  \item The labelling function shows the correspondence between the events and the
  elements of formula $\phi$, \ie{} $\forall t_i\in E: h(t_i)=v_i$, $\forall f_i\in
  E: h(f_i)=\setneg{v_i}$ and $\forall d_j\in E: h(d_j)=c_j$.
\end{itemize}
We now show that $\phi$ is satisfiable iff there exists an alternative $J$ to~$D$ after
$C \eqdef \emptyset$ in~$E$. This alternative is constructed by selecting for each
event $d_j \in D$ and event $e$ in immediate conflict. By construction of $\les_\phi$,
$h(e)$ is a literal in clause $h(d_j)=c_j$. Moreover, $C \cup J = J$ must be a configuration.
The causal closure is trivially satisfied since ${<} \eqdef \emptyset$. The conflict-freeness
implies that if $t_i\in J$ then $f_i\not\in J$ and vice-versa. Therefore, formula
$\phi$ is satisfiable iff an alternative~$J$ to~$D$ exists.

The construction of $\les_\phi$ is illustrated in~\figref{exsat} for:
$$\phi\eqdef
\underbrace{(x_1 \lor \setneg{x_2} \lor x_3)}_{c_1} \land \underbrace{(\setneg{x_1}
\lor \setneg{x_2})}_{c_2} \land \underbrace{(x_1 \lor \setneg{x_3})}_{c_3}$$

\begin{figure}[ht]
\centering

\begin{tikzpicture}[node distance=1cm,
	square/.style={regular polygon,regular polygon sides=4,inner sep=0,draw}]

\node at (-0.5,-2) [square,label=below:$t_1$] (A1) {$x_1$};
\node[square,label=below:$f_1$,right of=A1] (B1) {$\setneg{x_1}$};
\draw[dotted] (A1) -- (B1);
\node at (1.5,-2) [square,label=below:$t_2$] (A2) {$x_2$};
\node[square,label=below:$f_2$,right of=A2] (B2) {$\setneg{x_2}$};
\draw[dotted] (A2) -- (B2);
\node at (3.5,-2) [square,label=below:$t_3$] (A3) {$x_3$};
\node[square,label=below:$f_3$,right of=A3] (B3) {$\setneg{x_3}$};
\draw[dotted] (A3) -- (B3);

\begin{scope}[node distance=2cm]
\node[square,label=above:$d_1$] (D1) {$c_1$};
\draw[dotted] (D1) -- (A1);
\draw[dotted] (D1) -- (B2);
\draw[dotted] (D1) -- (A3);
\node[square,label=above:$d_2$,right of=D1] (D2) {$c_2$};
\draw[dotted] (D2) -- (B1);
\draw[dotted] (D2) -- (B2);
\node[square,label=above:$d_3$,right of=D2] (D3) {$c_3$};
\draw[dotted] (D3) -- (A1);
\draw[dotted] (D3) -- (B3);
\end{scope}

\end{tikzpicture}
\caption{Example of encoding a 3-SAT formula.}
\label{fig:exsat}
\end{figure}
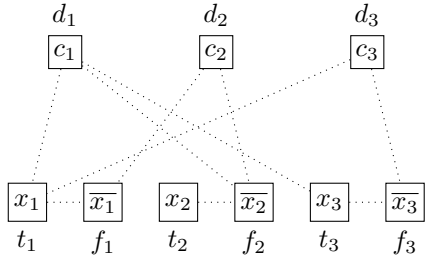
\end{proof}

\progaltnp*
\begin{proof}
Observe that the only difference between the statement of this theorem and that
of
\cref{thm:pes-altnp} is that here we assume the PES to be the unfolding of a
given program~$P$ under the relation~$\indep_P$.

As a result
the problem is obviously in NP, as restricting the class of PESs that we have as
input cannot make the problem more complex.

However, showing that the problem is NP-hard requires a new encoding, as the
(simple) encoding given for \cref{thm:pes-altnp} generates PESs that may not be
the unfolding of any program.
Recall that two events in the unfolding of a program are in
immediate conflict only if they are lock statements on the same variable.
So, in \cref{fig:exsat}, for instance, since $t_1 \cfl f_1$ and $f_1 \cfl d_2$, then
necessarily we should have $t_1 \cfl d_2$, as all the three events
should be locks to the same variable.

For this reason we give a new encoding of the 3-SAT problem into our problem.
As before,
let $V=\set{v_1, \ldots, v_n}$ be a set of Boolean variables.
Let $\phi \eqdef c_1 \land \ldots \land c_m$ be a 3-SAT formula, where
each clause $c_i \eqdef l_i \lor l'_i \lor l''_i$ comprises three literals.
A literal is either a Boolean variable $v_i$ or its negation $\setneg{v_i}$.
As before,
for a variable $v$, let $\posc v$ denote the set of clauses where~$v$ appears
positively and $\negc v$ the set of clauses where it appears negated.
We assume that every variable only appears either positively or negatively in a
clause (or does not appear at all), as clauses where a variable happens both
positively and negatively can be removed from~$\phi$.
As a result $\posc v \cap \negc v = \emptyset$ for every variable~$v$.

Let us define a program~$P_\phi$ as follows:
\begin{itemize}
\item For each Boolean variable $v_i$ we have two threads in~$P$, $t_i$
corresponding to $v_i$ (true), and $f_i$ corresponding to $\setneg{v_i}$ (false).
We also have one lock~$l_{v_i}$.
\item Immediately after starting, both threads $t_i$ and $f_i$ lock on~$l_{v_i}$.
This scheme corresponds to choosing a Boolean value for variable~$v_i$: the thread
that locks first chooses the value of~$v_i$.
\item For each clause $c_j \in \phi$, we have a thread $d_j$ and a lock $l_{c_j}$.
The thread contains only one statement which is locking $l_{c_j}$.
\item For each clause $c_j \in \posc{v_i} \cup \negc{v_i}$, the program contains one
thread $r_{\tup{v_i,c_j}}$ (run for variable $v_i$ in clause $c_j$). This thread contains
only one statement which is locking $l_{c_j}$.
\item After locking on $l_{v_i}$, thread $t_i$ starts in a loop all threads $r_{\tup
{v_i,c_j}}$, for $c_j \in \posc{v_i}$. Since we do not have thread creation in
our program model, we start a thread as follows: for each
thread~$r_{\tup{v_i,c_j}}$ we create an additional lock that is initially
acquired. Immediately after starting, $r_{\tup{v_i,c_j}}$ tries to acquire it.
When~$t_i$ wishes to start the thread, it just releases the lock, effectively
letting the thread start running.
\item Similarly, after locking on $l_{v_i}$, thread $f_i$ starts in a loop all threads
$r_{\tup{v_i,c_j}}$, for $c_j \in \negc{v_i}$.
\end{itemize}

When $P_\phi$ is unfolded, each statement of the program gives rise to exactly one
event in the unfolding. Indeed, by construction, each $t_i$ or $f_i$ thread starts
by a lock event and then causally lead to one $r$ event per clause the variable $v_i$
appears in. Any two of them concern different clauses and thus different locks,
and they are independent.

Let $C \eqdef \emptyset$ be an empty configuration, $D \eqdef \set{d_1, \ldots, d_m}$, and $U$ the set of all events in the unfolding of the program.

We now show that $\phi$ is satisfiable iff there exists an alternative $J$ to~$D$ after
$C \eqdef \emptyset$ in~$\unf{P_\phi,\indep_{P_\phi}}$.
This alternative is constructed by selecting for each
event $d_j \in D$ and event $e$ in immediate conflict. By construction of $P_\phi$,
it is a $r_{\tup{v_i,c_j}}$ where $v_i$ is a literal in clause $h(d_j)=c_j$. Moreover,
$C \cup J = J$ must be a configuration.
In order to satisfy the causal closure, since
\[
   {<} \eqdef
   \set{\tup{t_i,r_ {\tup{v_i,c_j}}}:c_j \in \posc{v_i}}
   \cup
   \set{\tup{f_i,r_{\tup{v_i,c_j}}}:c_j \in \negc{v_i}},
\]
$J$ must also contain the $t_i$ or $f_i$ preceding $r_{\tup{v_i,c_j}}$. The
conflict-freeness implies that if $t_i\in J$ then $f_i\not\in J$ and vice-versa.
Therefore, formula $\phi$ is satisfiable iff an alternative $J$ to~$D$ exists.

There are at most $2|V|+|\phi|(|V|+1)$ events, so the construction can be achieved
in polynomial time.
Therefore our problem is NP-hard.

The construction of $\unf{P_\phi}$ is illustrated in~\figref{progsat} for:
$$\phi\eqdef
\underbrace{(x_1 \lor \setneg{x_2} \lor x_3)}_{c_1} \land \underbrace{(\setneg{x_1}
\lor \setneg{x_2})}_{c_2} \land \underbrace{(x_1 \lor \setneg{x_3})}_{c_3}$$

\begin{figure}[ht]
\centering

\begin{tikzpicture}[node distance=1.5cm,
	square/.style={regular polygon,regular polygon sides=4,inner sep=0,draw}]

\node at (-3,-3) [square,label=below:$t_1$] (A1) {$x_1$};
\node[square,label=below:$f_1$,right of=A1] (B1) {$\setneg{x_1}$};
\draw[dotted] (A1) -- node[below]{$l_{x_1}$} (B1);
\node at (0,-3) [square,label=below:$t_2$] (A2) {$x_2$};
\node[square,label=below:$f_2$,right of=A2] (B2) {$\setneg{x_2}$};
\draw[dotted] (A2) -- node[below]{$l_{x_2}$} (B2);
\node at (3,-3) [square,label=below:$t_3$] (A3) {$x_3$};
\node[square,label=below:$f_3$,right of=A3] (B3) {$\setneg{x_3}$};
\draw[dotted] (A3) -- node[below]{$l_{x_3}$} (B3);

\node[square,label=below:$r_{\tup{x_1,c_3}}$,above of=A1] (R13) {\phantom{$c_1$}};
\draw[->] (A1) -- (R13);
\node[square,label=below:$r_{\tup{x_1,c_1}}$,left of=R13] (R11) {\phantom{$c_1$}};
\draw[->] (A1) -- (R11);
\node[square,label=below:$r_{\tup{x_1,c_2}}$,above of=B1] (R12) {\phantom{$c_1$}};
\draw[->] (B1) -- (R12);
\node[square,label=below:$r_{\tup{x_2,c_1}}$,above of=A2] (R21) {\phantom{$c_1$}};
\draw[->] (B2) -- (R21);
\node[square,label=below:$r_{\tup{x_2,c_2}}$,above of=B2] (R22) {\phantom{$c_1$}};
\draw[->] (B2) -- (R22);
\node[square,label=below:$r_{\tup{x_3,c_1}}$,above of=A3] (R31) {\phantom{$c_1$}};
\draw[->] (A3) -- (R31);
\node[square,label=below:$r_{\tup{x_3,c_3}}$,above of=B3] (R33) {\phantom{$c_1$}};
\draw[->] (B3) -- (R33);

\begin{scope}[node distance=2cm]
\node at (-2,0.5) [square,label=above:$d_1$] (D1) {$c_1$};
	\begin{scope}[red]
	\draw[dotted] (D1) -- node[left,very near start]{$l_{c_1}$} (R11);
	\draw[dotted] (D1) -- (R21);
	\draw[dotted] (D1) -- (R31);
	\draw[dotted] (R11) to [bend left=25] (R21);
	\draw[dotted] (R11) to [bend left=25] (R31);
	\draw[dotted] (R21) to [bend left=25] (R31);
	\end{scope}
	
\node[square,label=above:$d_2$,right of=D1] (D2) {$c_2$};
	\begin{scope}[blue]
	\draw[dotted] (D2) -- node[left,very near start]{$l_{c_2}$} (R12);
	\draw[dotted] (D2) -- (R22);
	\draw[dotted] (R12) to [bend left=25] (R22);
	\end{scope}
	
\node[square,label=above:$d_3$,right of=D2] (D3) {$c_3$};
	\begin{scope}[green!50!black]
	\draw[dotted] (D3) -- node[left,very near start]{$l_{c_3}$} (R13);
	\draw[dotted] (D3) -- (R33);
	\draw[dotted] (R13) to [bend left=25] (R33);
	\end{scope}
\end{scope}

\end{tikzpicture}
\caption{Program unfolding encoding a 3-SAT formula.}
\label{fig:progsat}
\end{figure}
\end{proof}

\thmcexnp*
\begin{proof}
   Given a Petri net $N \eqdef \tup{P, T, F, m_0}$,
	a transition $t \in T$,
   an independence relation ${\indep} \subseteq T \times T$,
   the unfolding $\les \eqdef \unf{M_N, \indep}$ of~$N$,
	and a configuration~$C$ of $\les$, we need to prove that
	deciding whether $h^{-1}(t) \cap \cex C = \emptyset$ is an NP-complete
   problem.

	We first prove that the problem is in~NP. This is achieved using a \emph{guess and
		check} non-deterministic algorithm to decide the problem.
	Let us non-deter\-min\-istically choose a configuration $C' \subseteq C$, in linear
	time
	on the input. A linearisation of~$C'$ is chosen and used to compute the marking~$m$ reached.
	We check that~$m$ enables~$t$ and that for any $<$-maximal event~$e$ of~$C$,
   $\lnot (h(e) \indep t)$ holds.
   Both tests can be done in polynomial time. If both tests succeed then we
   answer \emph{yes}, otherwise we answer \emph{no}.

	We now prove that the problem is NP-hard, by reduction from the 3-SAT problem.
	Let $V=\set{v_1, \ldots, v_n}$ be a set of Boolean variables.
	Let $\phi \eqdef c_1 \land \ldots \land c_m$ be a 3-SAT formula, where
	each clause $c_i \eqdef l_i \lor l'_i \lor l''_i$ comprises three literals.
	A literal is either a Boolean variable $v_i$ or its negation $\setneg{v_i}$.
	For a variable $v$, $\posc v$ denotes the set of clauses where~$v$ appears
	positively and $\negc v$ the set of clauses where it appears negated.
	
	Given $\phi$, we construct a $3$-safe net~$N_\phi$,
	an independence relation~$\indep$,
	a configuration~$C$ of the unfolding $\les$, and a transition~$t$ from~$N_\phi$
	such that $\phi$ is satisfiable iff some event in $\ex C$ is labelled by~$t$ :
	\begin{itemize}
		\item The net contains one place $d_i$ per clause $c_i$, initially empty.
		\item For each variable $v_i$ are two places $s_i$ and $s'_i$. Places $s_i$ initially
		contain $1$ token while places $s'_i$ are empty.
		\item For each variable $v_i$, a transition $p_i$ takes into account positive values
		of the variable. It takes a token from $s_i$, puts one in $s'_i$ (to move on to the
		other possibility for this variable) and puts one token in all places associated with
		clauses $c_j \in \posc {v_i}$. This transition mimics the validation of clauses where
		the variable appears as positive.
		\item For each variable $v_i$, a transition $n_i$ takes into account negative values
      of the variable. It takes a token from $s'_i$ and puts one token in all
      places associated
      with clauses $c_j \in \negc {v_i}$. It also removes one token from all
      places associated
		with clauses $c_j \in \posc {v_i}$, that have been marked by some $p_k$ transition.
      This transition $n_i$ mimics the validation of clauses where the variable
      appears as negative.
		\item Finally, a transition $t$ is added that takes a token from all $d_i$. Thus,
		it can only be fired when all clauses are satisfied, \ie{} formula $\phi$ is satisfied.
	\end{itemize}

 	The independence relation $\indep$ is the smallest binary, symmetric,
 	irreflexive relation such that
 	$p_i \indep p_j$ exactly when $i \ne j$ and
 	$p_i \indep n_j$ exactly when $i \ne j$.
 	Recall that $p_i, n_i$ correspond to respectively to the positive and negative valuations
 	of variable~$v_i$.
   In other words, the \emph{dependence} relation
   ${\depen} \eqdef T \times T \setminus {\indep}$ is the reflexive closure of the set
 	\[
 	\set{\tup{p_i,n_i} \colon 1 \le i \le n}
 	\cup
 	\set{\tup{t,p_i} \colon 1 \le i \le n}
 	\cup
 	\set{\tup{t,n_i} \colon 1 \le i \le n}
 	\]
 	Relation $\indep$ is an independence relation because:
 	\begin{itemize}
 		\item $\forall i\ne j$, transitions $p_i$ and $p_j$ do not share any input place ;
 		\item $\forall i\ne j$, the intersection between $\post{p_i}$ and $\pre{n_j}$ might
 		not be empty, but $n_j$ is always preceded by (and thus enabled after) $p_j$ (and
 		not $p_i$). So firing $p_i$ cannot enable, nor disable, $p_j$, and firing $p_i$
 		and $n_j$ in any order reaches the same state.
 	\end{itemize}
	
	Finally, configuration~$C$ contains exactly one event per $p_i$ and one per $n_i$,
	hence $2|V|$ events. This is because transition $n_i$ is dependent \emph{only}
	of $p_i$, and independent of (thus concurrent to) any other transition in~$C$.
	Thus formula $\phi$ has a model iff there is an event $e \in \en C$ labelled by~$t$.
	Indeed, initially only positive transitions $p_i$ are enabled that assign a positive
	value to their corresponding variable $v_i$. They add a token in all places $d_j$
	such that $c_j \in \posc{v_i}$. Then, when a negative transition $n_i$ fires, it deletes
	the tokens from these $d_j$ that had been created by $p_i$ since the variable cannot
	allow for validating these clauses anymore. It also adds tokens in the $d_k$ such
	that $c_k \in \negc{v_i}$ since the clauses involving $\setneg{v_i}$ now hold.
   Therefore,
	the number of tokens in a place $d_j$ is the number of variables (or their negation)
	that validate the associated clause. Formula $\phi$ is satisfied when all clauses
	hold at the same time, \ie{} each clause is validated by at least one variable. Thus
	all places $d$ must contain at least one token (and enable $t$) for $\phi$ satisfaction.
	
	The construction of $N_\phi$ is illustrated in~\figref{pnsat} for:
	$$\phi\eqdef
	\underbrace{(x_1 \lor \setneg{x_2} \lor x_3)}_{c_1} \land \underbrace{(\setneg{x_1}
		\lor \setneg{x_2})}_{c_2} \land \underbrace{(x_1 \lor \setneg{x_3})}_{c_3}$$
	
	\begin{figure}[ht]
		\centering

\begin{tikzpicture}[->,>=stealth,node distance=1cm,
	place/.style={circle,draw,inner sep=0pt,minimum size=4mm},
	transition/.style={rectangle,draw,inner sep=0pt,minimum size=3mm}]

\node[place,tokens=1,label=above:$s_1$] (S1) {};
\node[place,tokens=1,label=above:$s_2$,right of=S1,node distance=2cm] (S2) {};
\node[place,tokens=1,label=above:$s_3$,right of=S2,node distance=2cm] (S3) {};
\foreach \x in {1,2,3} {
	\node[transition,label=left:$p_\x$,below of = S\x] (P\x) {}
		edge [pre] (S\x);
	\node[place,label=left:$s'_\x$,below of=P\x] (SP\x) {}
		edge [pre] (P\x);
	\node[transition,label=left:$n_\x$,below of = SP\x] (N\x) {}
		edge [pre] (SP\x);
}
\foreach \x in {1,2,3}
	\node[place,label=left:$d_\x$,below of=N\x] (D\x) {};
\node[transition,label=below:$t$,below of=D2] (T) {}
	edge [pre] (D1)
	edge [pre] (D2)
	edge [pre] (D3);

\draw (P1) to [bend right=45] (D1);
\draw (P1) -- (D3);
\draw (D1) -- (N1);
\draw (N1) -- (D2);
\draw (D3) -- (N1);
\draw (N2) -- (D1);
\draw (N2) -- (D2);
\draw (P3) -- (D1);
\draw (D1) -- (N3);
\draw (N3) -- (D3);

\end{tikzpicture}
		\caption{Petri Net encoding a 3-SAT formula.}
		\label{fig:pnsat}
	\end{figure}
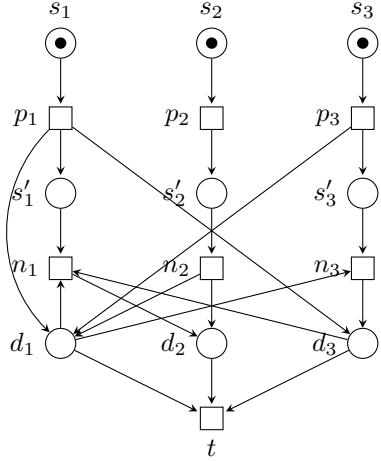
	
\end{proof}

\section{Proofs for Causality Trees}

\procaus*

\begin{proof}
Firstly, we show that $e < e'$ holds iff $e = \tmax{e',i} \lor e < \tmax{e', i}$.
\begin{itemize}
\item
   Direction $\Rightarrow$.
	Assume that $e < e'$.
   This implies that $e \in \lceil e'\rceil$ and there must exist 
	$\hat e \in [e']$ such that $\hat e= tmax(e',i)$.
   Since both~$e$ and~$\hat e$ are events from thread~$i$, and both are
   contained in~$[e']$ they cannot be in conflict, but
   $\lnot (h(e) \indep h(\hat e))$.
   Then either $e = \hat e$ or $e < \hat e$.
\item
   Direction $\Leftarrow$.
   Let $\hat e \eqdef \tmax{e',i}$.
   Since $i \ne i'$ we have that $\hat e \ne e'$,
   and since $\hat e \in [e']$ we have that $\hat e < e'$.
	Let $e \in \les$ be any event such that either $e = \hat e$ or $e < \hat e$.
   We then have $e \leqslant \hat e < e'$, so clearly $e < e'$.
\end{itemize}

Now we show that
$e \cfl e'$ holds iff there is some $l \in \locks$ such that
$\lmax{e, l} \cfl \lmax{e', l}$.

\begin{itemize}
\item
   Direction $\Rightarrow$.
   Assume that $e \cfl e'$ holds.
   Then necessary there exist events
   $e_1' \in [e]$ and $e_2' \in [e']$ such that $e_1' \icfl e_2'$. 
   Since only lock events touching the same variable are able 
   to create immediate conflicts, we obviously know that $\exists l \in \locks:
   h(e_1') = h(e_2') = \tup{\lock,l}$.
   Since $e_1'\in [e]$ then $\exists e_1 \in [e]: \exists e_1 = \lmax{e,l}$.
   Similarly, $\exists e_2 \in [e']: e_2 = \lmax{e',l}$.
   Both $e_1$ and $e_2$ are $<$-maximal events, so $e_1' < e_1$ or $e_1'
   =e_1$ and $e_2 < e_2'$ or $e_2 = e_2'$. The conflict is inherited, having
   $e_1' \icfl e_2'$ implies $e_1 \cfl e_2$.
\item
   Direction $\Leftarrow$.
   Assume that there is some $l \in \locks$ such that
   $\lmax{e,l} \cfl \lmax{e',l}$ and
   let $e_1 \in [e]: e_1 = \lmax{e,l}$ and $e_2 \in [e']: e_2 = \lmax{e',l}$,
   then $e_1 \cfl e_2$.
   Since $e_1 \in [e]$, we have $e_1 < [e]$. Similarly, $e_2 \in [e']$,
   \ie, $e_2 < e'$.
   The conflict is inherited and $e_1 \cfl e_2$, so necessarily $e \cfl e'$.
\end{itemize}	
\end{proof}

\section{Experiments with the SV-COMP'17 Benchmarks}
\label{sec:svcomp17}

In this section we present additional experimental results using the SV-COMP'17
benchmarks. In particular we use the benchmarks from the \texttt{pthread/}
folder.\footnote{See
\url{https://github.com/sosy-lab/sv-benchmarks/releases/tag/svcomp17}.}

\newcommand\newrow{\\[0.0pt]}
\newcommand\param[1]{\scriptsize(#1)}

\newcommand\cmidrules{
  \cmidrule(r){1-1}
  \cmidrule(r){2-3}
  \cmidrule(r){4-5}
}

\begin{table}[!bt]
\footnotesize

\setlength\tabcolsep{3pt}
\def\sep{\hspace{25pt}}
\def\tinysep{\hspace{5pt}}
\centering
\tt
\begin{tabular}[t]{l@{\sep}rr@{\sep}rr}	
\toprule
  \multicolumn{1}{l}{\rm Benchmark}
& \multicolumn{2}{l}{\rm \dpu (k=1)}  
& \multicolumn{2}{l}{\rm \nidhugg}

\newrow
\cmidrules

  {\rm Name}

& {\rm Time}
& {\rm Bug}

& {\rm Time}
& {\rm Bug}

\newrow
\midrule

\rm\tt bigshot-p-false                       &    0.46  &     y &   0.20  &    y \newrow           
\rm\tt bigshot-s2-true                       &    0.45  &     n &   0.20  &    n \newrow           
\rm\tt bigshot-s-true                        &    0.45  &     n &   0.18  &    n \newrow           
\rm\tt fib-bench-false                       &    0.87  &     y &   0.69  &    y \newrow           
\rm\tt fib-bench-longer-false                &    2.57  &     y &   1.57  &    y \newrow           
\rm\tt fib-bench-longer-true                 &    2.23  &     n &   2.75  &    n \newrow           
\rm\tt fib-bench-longest-false               &       TO &       &      TO &      \newrow           
\rm\tt fib-bench-longest-true                &       TO &       &      TO &      \newrow           
\rm\tt fib-bench-true                        &    0.89  &     n &   0.76  &    n \newrow           
\rm\tt indexer-true                          &       TO &       &      TO &      \newrow           
\rm\tt lazy01-false                          &    0.42  &     y &   0.82  &    y \newrow           
\rm\tt queue-false                           &    0.70  &     y &   0.21  &    y \newrow           
\rm\tt queue-longer-false                    &    0.96  &     y &   0.53  &    y \newrow           
\rm\tt queue-longest-false                   &    1.80  &     y &   0.53  &    y \newrow           
\rm\tt queue-ok-longer-true                  &    0.44  &     n &   0.29  &    n \newrow           
\rm\tt queue-ok-longest-true                 &    0.46  &     n &   0.37  &    n \newrow           
\rm\tt queue-ok-true                         &    0.49  &     n &   0.19  &    n \newrow           
\rm\tt sigma-false                           &    0.30  &     y &   0.24  &    y \newrow 
\rm\tt singleton-false                       &    0.48  &     y &   0.21  &    y \newrow           
\rm\tt singleton-with-uninit-problems-true   &    0.47  &     n &   0.20  &    n \newrow 
\rm\tt stack-false                           &    0.66  &     y &   0.21  &    y \newrow           
\rm\tt stack-longer-false                    &    0.94  &     y &   1.50  &    y \newrow           
\rm\tt stack-longer-true                     &       TO &       &      TO &      \newrow           
\rm\tt stack-longest-false                   &    1.85  &     y &   4.48  &    y \newrow           
\rm\tt stack-longest-true                    &       TO &       &      TO &      \newrow           
\rm\tt stack-true                            &    0.52  &     n &   0.35  &    n \newrow           
\rm\tt stateful01-false                      &    0.44  &     y &   0.20  &    y \newrow           
\rm\tt stateful01-true                       &    0.44  &     n &   0.19  &    n \newrow           
\rm\tt twostage-3-false                      &    0.48  &     y &   0.40  &    y \newrow

\bottomrule
\end{tabular}
\vspace{4pt}
\caption{\footnotesize Comparing \dpu and \nidhugg
on the \texttt{pthread/} folder of the SV-COMP'17 benchmarks.
Machine: Linux, Intel Xeon 2.4GHz.
TO: timeout after 8 min.
Columns are:
Time in seconds,
Bug: \texttt{y} if the bug is detected, \texttt{n} if no bug is detected.}

\label{tab:svcomp17}
\end{table}

\let\newrow\undefined
\let\param\undefined
\let\cmidrules\undefined

All benchmarks were taken from the official repository of the SV-COMP'17.
We modified almost all of them to remove the dataraces, using one or more
additional mutexes.
All benchmarks have between~50 and~170 lines of code.
Most of them employ~2 or~3 threads but some of them reach up to~7 threads.

The first remark is that both tools correctly classified every benchmark
as \emph{buggy} or \emph{safe}.
In \dpu we used QPOR with $k=1$ and the exploration was optimal on all
benchmarks. That means that \nidhugg and \dpu are doing a very similar
exploration of the statespace in these benchmarks.
As a result, it is not surprising that both tools timeout on exactly the same
benchmarks (5 out of 29).
On the other hand most benchmarks in this suite are quite simple for DPOR
techniques: the longest run time for \dpu was 2.6s (and 4.5s for \nidhugg).

In general the run times for \nidhugg are slighly better than those of
\dpu. We traced this down to two factors.
First, while \dpu is in general faster at exploring new program interleavings,
it has a slower startup time.
Second, when \dpu finds a bug, it does not stop and report it,
it continues exploring the state space of the program.
This is in contrast to~\nidhugg, which stops on the
first bug found. We will obviously implement a new mode in~\dpu where the tool
stops on the first bug found, but for the time being this visibly affects \dpu on
bechmarks such as the \texttt{queue-*-false}, where~\nidhugg is almost twice
faster than~\dpu.

\end{document}